\documentclass[opre,nonblindrev]{informs3} 

\OneAndAHalfSpacedXI 

\usepackage{endnotes}

\usepackage{natbib}
 \bibpunct[, ]{(}{)}{,}{a}{}{,}%
 \def\bibfont{\small}%
\usepackage{booktabs} 

\usepackage[ruled]{algorithm2e} 

\usepackage{ bbm, color, enumerate, amsbsy, amsfonts, amsopn, amssymb,  amsxtra, bezier, graphicx, latexsym, verbatim,
	pictexwd, supertabular, url, dsfont,leftidx, setspace}
\usepackage{mathtools,bm}
\usepackage{multirow}
\usepackage{footnote}
\makesavenoteenv{tabular}
\makesavenoteenv{table}

\usepackage[ruled]{algorithm2e}

\newcommand{\mb}[1]{\ensuremath{\boldsymbol{#1}}}
\newcommand{\one}{\mathbbm{1}^\omega}

\newcommand{\onee}{\mathbbm{1}}

\def\opt{\textsc{OPT}}
\def\alg{\textsc{ALG}}
\def\topt{\textsc{OPT}^{tr}}

\def\umi{u_{\text{-}i}}
\def\omitt{\omega_{\text{-}it}}
\def\ui{u_{i}}
\def\ub{\mb{u}}
\def\mcala{\mathcal{A}}

\TheoremsNumberedThrough     
\ECRepeatTheorems

\EquationsNumberedThrough    

\begin{document}


\RUNAUTHOR{Goyal and Udwani}

\RUNTITLE{Online Matching with Stochastic Rewards}
\TITLE{Online Matching with Stochastic Rewards: Optimal Competitive Ratio via Path Based Formulation}

\ARTICLEAUTHORS{%
\AUTHOR{Vineet Goyal}
\AFF{Columbia University,
	 \EMAIL{vgoyal@ieor.columbia.edu}} 
\AUTHOR{	Rajan Udwani}
\AFF{UC Berkeley,
	 \EMAIL{rudwani@berkeley.edu}}
} 

\ABSTRACT{%
The problem of online matching with stochastic rewards is a generalization of the online bipartite matching problem where each edge has a probability of success. When a match is made it succeeds with the probability of the corresponding edge. Introducing this model, \cite{deb} focused on the special case of identical edge probabilities. Comparing against a deterministic offline LP, they showed that the Ranking algorithm of Karp et al. (STOC 1990) is 0.534 competitive and proposed a new online algorithm with an improved guarantee of $0.567$ for vanishingly small probabilities. For the case of vanishingly small but heterogeneous probabilities, \cite{mehta} gave a 0.534 competitive algorithm against the same LP benchmark.
 
We consider the more general vertex-weighted version of the problem and give two new results. First, we show that a natural generalization of the Perturbed-Greedy algorithm of \cite{goel}, is $(1-1/e)$ competitive when probabilities decompose as a product of two factors, one corresponding to each vertex of the edge. This is the best achievable guarantee as it includes the case of identical probabilities and in particular, the classical online bipartite matching problem. Second, we give a deterministic $0.596$ competitive algorithm for the previously well studied case of fully heterogeneous but vanishingly small edge probabilities. A key contribution of our approach is the use of novel path-based formulations and a generalization of the primal-dual scheme of \cite{devanur}. These allow us to compare against the natural benchmarks of adaptive offline algorithms that know the sequence of arrivals and the edge probabilities in advance, but not the outcomes of potential matches. 
These ideas may be of independent interest in other online settings with post-allocation stochasticity.

}%



\maketitle

%


\section{Introduction}\label{intro}	
Online bipartite matching and its variants and generalizations have been intensely studied for past few decades \citep{survey}. Internet advertising is a major application domain for these problems, along with Crowdsourcing \citep{crowd1,crowd2}, 
personalized recommendations \citep{negin, will}, ridesharing~\citep{ashlagi,ali,rad}, online order order fulfillment \citep{fultut} and many others. Despite a rich body of work on these problems, there are several basic questions that remain open. In this work, we study one such problem. To introduce our setting, we start with the classic online bipartite matching problem.
\smallskip

\noindent \textbf{Online bipartite matching~\citep{kvv}:} 
Consider a graph $G(I,T,E)$ where the vertices in $I$, called \emph{ resources} (may correspond to advertisers in Internet advertising and tasks in Crowdsourcing), are known in advance. Vertices in $T$, called \emph{arrivals} (corresponding to customers, ad slots, workers), are sequentially revealed one at a time. When a vertex $t\in T$ arrives, the set of edges $(i,t)\in E$ is revealed. Without loss of generality (w.l.o.g.), the vertices in $T$ arrive in the order of their index. 
After each arrival, the (online) algorithm must make an irrevocable decision to offer at most one (available/unmatched) vertex $i$ with an edge to $t$. Unmatched arrivals depart immediately and cannot be matched later. The goal is to maximize the total number of matches and the performance of online algorithm is compared against an optimal offline algorithm that knows the entire graph $G$ in advance. Formally, let 
$OPT(G)$ denote the optimal achievable by an offline algorithm and ${A}(G)$ the expected value achieved by a (possibly randomized) online algorithm $A$. The goal is to design an algorithm that maximizes the competitive ratio, 
\[\min_{G}\frac{A(G)}{OPT(G)}.\]

 \cite{kvv} proposed and analyzed the Ranking algorithm that orders all offline vertices in a random permutation and matches every online vertex to the highest ranked neighbor available. They showed that Ranking attains the best possible competitive ratio of $(1-1/e)$ for online bipartite matching (see \cite{baum,goel2,devanur} for simplified analysis). 

In many applications, matches may fail with non-zero probability and the success or failure of a match is only revealed after the match has been made. This \emph{post-allocation stochasticity} arises, for instance, in Internet advertising where with some probability a user may not click on the ads shown. In the popular \emph{pay-per-click} model, an advertiser pays only if the user actually clicks i.e., the match succeeds. Similarly, in personalized recommendations new revenue is generated only if the customer buys an offered product. 
\smallskip

\noindent \textbf{Online matching with stochastic rewards:} Motivated by the observations above, \cite{deb} introduced a generalization of online bipartite matching where the objective is to maximize the expected number of \emph{successful} matches. Each edge $(i,t)$ {\color{black} has a success probability} $p_{it}$ that is revealed online along with the edge. When an algorithm makes a match, the arrival accepts (\emph{successful} match or reward) with probability given by the corresponding edge. The success of a given match is independent of past ones and the outcome of the stochastic reward is revealed only after the match is made. An unsuccessfully matched arrival departs immediately, leaving the resource available for future re-match. 
A natural generalization of this problem allows \emph{vertex-weighted rewards} where each resource $i\in I$ has a (possibly different) reward $r_i$ \footnote{If rewards are allowed to depend on both the arrival and resource the problem becomes much harder. It was shown in \cite{goel} that in this case no algorithm with competitive ratio better than $O\big(\frac{1}{|T|}\big)$ is possible.}. This is the setting we study.

The introduction of stochastic rewards 
raises an interesting question regarding the nature of the offline benchmark to compare against online algorithms. In addition to knowing all edges and their probabilities in advance, should the offline algorithm know the realizations of the rewards too? Unsurprisingly, such a benchmark turns out to be too strong for any meaningful bound\footnote{Consider a single arrival with an edge to $n\geq 1$ identical resources. Each edge has a probability of success $1/n$. A benchmark that knows realizations of rewards attains expected reward at least $(1-1/e)$. However, no online algorithm can achieve a reward better than $1/n$.}. {\color{black} A natural alternative} is to compare against adaptive offline algorithms that know the arrivals, edges and edge probabilities in advance but must sequentially adapt to the outcomes of matches. In particular, rewards are realized only after matching decisions are made and offline can match any arrival at most once. Though offline knows all arrivals in advance, the sequence in which offline makes matching decisions gives rise to different benchmarks.
\begin{itemize}
	\item  \textbf{Clairvoyant}: An often used benchmark where the offline algorithm makes matching decisions in order of the arrival sequence i.e., vertex $t\in T$ is matched only after matching decisions for all vertices $t'<t$ have been made 
	 \citep{negin,will,chan,reuse,feng}.
	\item \textbf{Fully Offline}: Algorithm can adaptively choose an order for matching vertices in $T$, given the realized outcomes of previous matches \citep{brubach2,toronto}. 
\end{itemize}

Note that in the deterministic case (all probabilities are one) these benchmarks are identical since the order in which vertices are matched does not matter given knowledge of the entire arrival sequence. In general, there are merits to using either benchmark. On one hand, fully offline is a stronger benchmark than clairvoyant. On the other hand, clairvoyant benchmark generalizes quite naturally to other settings with post-allocation stochasticity. For instance, in case of \emph{reusable} resources where successfully matched resources can be re-matched after some stochastic duration \citep{reuse}, allowing offline to match in an order different from the order of arrivals does not lead to a meaningful benchmark. Due to this generalizability, our main focus in this paper is to compare against clairvoyant benchmark. Interestingly, we find that our results hold more strongly against fully offline as well. For simplicity, we refer to clairvoyant and fully offline collectively as \emph{offline}.

 Unfortunately, the performance of offline can be hard to characterize exactly. However, observe that an upper bound on offline is easily given by the optimum of the following LP (see \cite{will,negin} for a formal proof). 
\begin{eqnarray}
OPT(LP)=	&\max &\sum_{(i,t)\in E} p_{it}r_ix_{it}\nonumber \\
&s.t.\ &\sum_{t\in T} p_{it} x_{it} \leq 1\quad \forall i\in I \label{expectedconserve}\\
&& \sum_{i\in I} x_{it} \leq 1 \quad \forall t\in T\label{conserve}\\
&& 0\leq x_{it}\leq 1 \quad \forall (i,t)\in E,	\nonumber
\end{eqnarray} 
here $x_{it}$ is a decision variable corresponding to the fraction of resource $i$ matched to arrival $t$. Constraints \eqref{expectedconserve} enforce that the capacity constraint is satisfied in expectation, for every resource. We let the capacity of each resource be 1 w.l.o.g..\footnote{Results for the case of unit capacity of each resource generalize to arbitrary capacities. Given a setting with inventory $c_i$ for resource $i$, consider a unit inventory setting where in place of resource $i$ we now have $c_i$ resources each with inventory of 1 and arrivals that have edges to all $c_i$ vertices for every edge $(i,t)$ in the original instance. Offline algorithm knows all arrivals in advance and thus knows these copies represent the same resource. Therefore, the optimal offline value remains unchanged and an algorithm for the unit inventory case can be used for arbitrary inventory levels without loss in guarantee. } Similarly, constraints \eqref{conserve} enforce that the matched fractions of every arrival sum to at most 1. We call this the \emph{expectation} LP as the capacity constraints need only hold in expectation. Due to the simplicity of this LP, when evaluating competitive ratios against offline one often chooses to use the upper bound $OPT(LP)$.  
\smallskip

\noindent \textbf{Previous work:} \cite{negin} and \cite{mehta} showed that it is possible to obtain $1/2$ competitive guarantee against the expectation LP benchmark 
using {\color{black} deterministic greedy policies}. 
Finding an algorithm that is provably better than na\"ive greedy is typically the main objective in this class of online allocation problems~\citep{msvv, deb, survey}. 
This has proven to be a major challenge for the stochastic rewards setting. 
Progress has been made in special cases with remarkable new insights. 
In applications such as online-advertisement, the edge probabilities or \emph{click-through-rates} are sometimes quite small \citep{survey}. Consequently, this is a previously well studied case for the stochastic rewards problem. In particular, \cite{deb} found that for identical rewards and identical probabilities i.e., $r_i=1,\forall i\in I$ and $p_{it}=p\, \forall (i,t)\in E$: Assigning arrivals to available neighbors with the least number of failed matches in the past (called StochasticBalance) is 0.567 competitive as $p\to 0$. They also showed that Ranking is 0.534 approximate for arbitrary value of $p$. Complementing these lower bounds, they showed that no algorithm has a competitive ratio better than $0.621<(1-1/e)$ against the LP. This indicates that perhaps the problem is strictly harder than the classical setting of \cite{kvv}, where a $(1-1/e)$ guarantee is possible. 
\cite{mehta} studied the case of heterogeneous and vanishingly small edge probabilities and gave a 0.534 competitive algorithm (called SemiAdaptive). No results beating $1/2$ are known to us for other cases, or when the rewards $r_i$ are not identical, even for identical probabilities.

\subsection{Our Results}\label{ressum}
Our first result establishes the best possible guarantee of $(1-1/e)$ for a special case of online matching with stochastic rewards where every edge probability is \emph{decomposable} i.e., 
is a product of two terms corresponding to each of the two vertices\footnote{Note that similar special cases have been considered in other online matching settings.  
	For instance, \cite{charikar} consider a deterministic online matching problem where each edge has a weight $w_{it}$ and the goal is find a matching with maximum total weight. They consider the special case where $w_{it}=w_i\times w_{t},\, \forall (i,t)\in E$, and give a $(1-1/e)$ competitive algorithm for this setting.}.
\begin{theorem}\label{maint}
	For decomposable probabilities $p_{it}=p_ip_t\, \forall (i,t)\in E$, Algorithm \ref{rank} is $(1-1/e)$ competitive against offline (both clairvoyant and fully offline). This is the best possible competitive ratio guarantee.
\end{theorem}
Algorithm \ref{rank} computes a randomly perturbed per-unit reward for every resource and makes each match with the goal of maximizing the expected perturbed reward from the match. The algorithm is formally presented in Section \ref{sec:decomp}. Observe that the case of decomposable probabilities (with otherwise arbitrary magnitudes) includes the case of identical probabilities $p_{it}=p\, \forall (i,t)\in E$, for which \cite{deb} showed that there is no $0.621<(1-1/e)$ competitive algorithm when comparing with $OPT(LP)$. Theorem \ref{maint} establishes that this upper bound is largely due to the limitations of the expectation LP benchmark (see Section \ref{sec:badLP}). 
To obtain a tighter bound on offline we introduce a novel and intuitive path based program. In contrast to the expectation LP, this program has a constraint for every possible sample path instance of the stochastic reward process, along with constraints that enforce the independence of offline's decisions from the outcome of stochastic reward. This allows us to achieve the best possible guarantee of $(1-1/e)$ with a natural generalization of an algorithm from deterministic online matching. 


Let $p=\max_{(i,t)\in E} p_{it}$, denote the maximum edge probability in a given instance. For general (not necessarily) decomposable probabilities, we establish the following guarantee. 
\begin{theorem}\label{mainh}
There exists a constant $\eta>0$ such that	Algorithm \ref{fullada} is $(1-\eta\,\sqrt{p})\, 0.596$ competitive against offline for $p\leq 0.38$. 
\end{theorem}
Similar to Algorithm \ref{rank}, Algorithm \ref{fullada} matches each arrival greedily based on perturbed rewards. However, the perturbations are deterministic and change over time. The algorithm is formally presented in Section \ref{sec:vanish}. Observe that, for instances with vanishingly small edge probabilities i.e., $p\to 0$, the theorem implies a guarantee of 0.596. 
Recall that the importance of the special case of small probabilities stems from the observation that in applications such as online-advertisement, the probabilities or \emph{click-through-rates} are often quite small \citep{survey}. 
The path based program that we introduce to prove Theorem \ref{maint} appears to be insufficient for analyzing non-decomposable probabilities even when they are vanishingly small. To show Theorem \ref{mainh}, we develop a novel LP free analysis technique that bypasses the need for an LP based upper bound and uses path-based arguments to directly compare online algorithms with offline. 

\begin{table}
	\centering
	\begin{tabular}{|l|l|l|l|}
		\cline{2-4}
		\multicolumn{1}{l|}{}   & \multirow{2}{*}{Identical rewards} & \multicolumn{2}{c|}{ Vertex-weighted rewards}                                                                       \\
		\cline{3-4}
		\multicolumn{1}{c|}{}  	&                                    & $c_i=1$ & \multicolumn{1}{c|}{$c_i\to \infty$}  \\ 
		\hline
		$p_{it}=p\to 0$ & $0.567 \to  \pmb{(1-1/e)}$      & \multirow{2}{*}{$0.5 \to \pmb{(1-1/e)}$}  & \multirow{4}{*} {$(1-1/e) $}\\
		\cline{1-2}                              
		$p_{it}=p_i p_t$              & $0.500\to \pmb{(1-1/e)}$      &  & \\     
		\cline{1-3}                                  
		$p_{it}\to 0$   & $0.534 \to  \pmb{0.596}$                                &   $0.5 \to  \pmb{0.596}$                         &  \\
		\cline{1-3}                                
		General $p_{it}$               & 1/2 
		& 1/2 
		&      \\		\hline                            
	\end{tabular}
	\caption{Summary of best known competitive ratio guarantees in relevant settings. Here $c_i$ represents the number of copies of each vertex $i\in I$ available, with $c_i=1$ being the most general case. Entries in bold represent our contributions. }
	\label{tab}
\end{table}

A well studied regime for online resource allocation problems, including the stochastic rewards setting, is when we have a large number of copies $c_i$ for every vertex $i\in I$. 
In this \emph{large capacity regime}, previous results show that one can achieve the best possible $(1-1/e)$ guarantee against expectation LP \citep{negin,msvv}. \emph{Why does the expectation LP benchmark admit this guarantee only for large inventory?} 
{\color{black} To answer this question, 
in Section \ref{asymptote} we show that when $c_i\to +\infty\, \, \forall i\in I$, the expectation LP and offline benchmark converge to the same optimal value.} 

\subsection{Related Work}\label{related}
We briefly review the vast body of literature on online resource allocation, including work we discussed previously to put things into perspective. \cite{kvv} introduced the online bipartite matching problem and proposed the optimal $(1-1/e)$ competitive Ranking algorithm. \cite{baum}, \cite{goel2} considerably simplified (and completed) the original analysis. Subsequently, \cite{goel}, proposed the Perturbed-Greedy algorithm and showed that it is $(1-1/e)$ competitive more generally for the case of online vertex-weighted bipartite matching. \cite{devanur} gave an elegant and intuitive randomized primal-dual framework for proving these results. Their framework applies to and simplifies several other settings, such the related AdWords problem of \cite{msvv} (and the earlier result for $b$-matchings by \cite{pruhs}). There have also been a series of results on different types of stochastic arrival models, where there is distributional information in the arrivals that can be exploited for better results (for example, \cite{devhay}), 
as well as simultaneous guarantees in adversarial and stochastic arrival settings (for instance, \cite{mirrokni}). For a detailed survey of these models and results we refer the reader to a monograph by \cite{survey}. 

Moving on to settings with stochastic rewards, \cite{deb} introduced the model we consider here and focused on the case where all rewards and edge probabilities are identical. They showed the first results beating $1/2$ (and achieving 0.534) for this setting and motivated the case of small edge probabilities. \cite{negin}, 
considered a broad generalization of stochastic rewards where one offers \emph{a set or an assortment of resources} to each arrival. The arrival then chooses at most one resource from this set based on a \emph{choice model} that is revealed online at the time of arrival. With the objective of maximizing total expected reward, they showed that when the number of copies (inventory) of each resource approaches $+\infty$, an \emph{inventory balancing} algorithm\footnote{A generalization of the algorithm for AdWords in \cite{msvv}} is asymptotically $(1-1/e)$ competitive. For the most general case of unit inventory, 
their algorithm reduces to the Greedy algorithm and is $1/2$ competitive. Of course, this also establishes a $1/2$ competitive result for stochastic rewards with arbitrary vertex-weights and edge probabilities. 
The $1/2$ competitiveness of a non-adaptive Greedy algorithm for identical rewards was also shown independently by \cite{mehta}. 
They also gave a 0.534 competitive algorithm for identical rewards and heterogeneous but vanishingly small edge probabilities. Note that in contrast to \cite{deb} and \cite{mehta}, \cite{negin} evaluate competitiveness against the clairvoyant benchmark. 
However, for the sake of analysis they resort to an upper bound on the performance of clairvoyant and use the expectation LP. Clairvoyant algorithms are a commonly used benchmark in various other settings that include the element of stochastic rewards and more generally, other forms of post-allocation stochasticity (for instance, see \cite{will,chan,reuse,reuseopt}).

Since the online/arXiv version of this paper first appeared, there have been further developments of interest. In concurrent work, \cite{brubach2} consider at a 
setting where for each arrival a finite number of \emph{rematches} are allowed in case of unsuccessful matches. They observe that in this case the problem for a single arrival is optimally solved by a dynamic program (DP). Using this insight show that the DP based greedy algorithm is 0.5 competitive. Further, they highlight several limitations of the LP benchmark and advocate the use of fully offline as a more natural benchmark for future work. Subsequently, \cite{huang} 
gave a 0.572 competitive for heterogeneous and strongly vanishing probabilities ($p=1/f(n)$, where $f(n)$ is a strictly increasing function of the number of resources $n$). Their result is shown against the expectation LP benchmark. Thus, while their benchmark is stronger than offline the stronger assumption on smallness of probabilities and different numerical guarantee (0.572 in \cite{huang} vs.\ $1-1/e$ and $0.596$ here), renders their results incomparable with ours. 
Very recently, \cite{unkads} analyzed a variant of Algorithm \ref{rank} for an online allocation setting with unknown resource capacities. A byproduct of his result is that Algorithm \ref{rank} is at least 0.522 competitive for stochastic rewards with vanishing edge probabilities. 

There is also a stream of work that looks at the aspect of stochastic rewards with stochastic arrivals and other related (including offline) versions of the problem. In particular, \cite{bansal} were the first to study stochastic rewards in the context of bipartite matching. They considered an offline \emph{probe-commit} version of the problem where all edges and probabilities are known in advance and one must probe edges with the restriction that if a probed edge exists, it must be included in the matching. 
We refer the reader to \cite{brubach2} and \cite{toronto} for a more detailed review of these settings.

\subsubsection*{Outline for rest of the paper:} 
We start by introducing notation in Section \ref{sec:prelim}. In Section \ref{sec:decomp}, we consider the case of decomposable probabilities and present Algorithm \ref{rank}. Section \ref{sec:badLP} discusses the challenges with using the randomized primal-dual framework of \cite{devanur} for the stochastic rewards problem. This motivates the consideration of our path based analysis in Section \ref{sec:pbp} and Section \ref{sec:pbdf}. In Section \ref{sec:vanish}, we consider the case of small (not necessarily decomposable) probabilities. We start by presenting Algorithm \ref{fullada}, followed by the proof of Theorem \ref{mainh} (Section \ref{sec:vanish1} to Section \ref{sec:fsb}). 
Section \ref{asymptote} completes the picture with a convergence result in the large inventory regime. Finally, in Section \ref{open} we conclude with a review of some relevant open problems. 

\section{Preliminaries} \label{sec:prelim}
Recall the problem definition. We have a bipartite graph $G$ with a set $I$ of $n$ offline vertices, also called resources. An arbitrary number of vertices $t\in T$ arrive online. We use the index of arrivals $t$ to also denote their order in time. So assume vertex $t\in T$ arrives at time $t$. All edges $(i,t)$ incident on vertex $t$ are revealed when $t$ arrives, along with a corresponding probability of success $p_{it}$. On each new arrival, we must make an irrevocable decision to match the arrival to any one of the available offline neighbors. Once a match, say $(i,t)$, is made, it succeeds with probability $p_{it}$ independent of all other events. If the match succeeds, resource $i$ unavailable to future arrivals and leads to a reward $r_i$. The objective is to maximize the expected reward summed over all arrivals. Throughout the paper, we use \alg\ to denote the online algorithm under consideration as well its expected total reward. Similarly, \opt\ denotes the (fully) offline benchmark and its expected total reward. W.l.o.g., let \opt\ be a deterministic algorithm. 
\begin{lemma}\label{determ}
	The optimal fully offline and clairvoyant algorithms are deterministic and given by a (possibly exponential time) dynamic program.
\end{lemma} 
We provide a proof in Appendix \ref{appx:trunc}. While the ideas in this paper can be used for comparison with randomized benchmark algorithms, the fact that \opt\ is deterministic simplifies presentation of the key ideas.

For edge $(i,t)\in E$, let $\onee(i,t)$ be an indicator random variable that denotes the success/failure of the edge. So $\onee(i,t)$ equals 1 w.p.\ $p_{it}$. 
{\color{black} Recall that, the success/failure of an edge is independent of the outcome of other edges. Let $\omega$ denote a sample path over the success/failure of all edges. Formally, $\omega$ defines for every edge $(i,t)\in E$, the value $\onee(i,t)=\one(i,t)$ for the success/failure of $(i,t)$. Conditioned on $\omega$, any given (online or offline) algorithm observes the outcome $\one(i,t)$ after matching $i$ to $t$. Outcomes for edges not included in the matching are not revealed to the algorithm. } 
Let $\Omega$ denote the universe of all sample paths $\omega$.   For simplicity, we use $E_{\omega}[\cdot]$ to denote expectation over the randomness in outcomes of all edges.  
We use $\omitt$ to represent part of sample path $\omega$ that defines the success/failure of all edges except edge $(i,t)\in E$. In other words, $\omitt$ determines values of random variables $\onee(j,t')$, for all edges $(j,t')\neq (i,t)$. Overloading notation, we let $E_{\omega}[\cdot\mid \omitt]$ denote expectation over randomness in outcomes of edge $(i,t)$, with the outcomes of all other edges fixed according to $\omitt$. Let $E_{\omitt}[\cdot]$ denote expectation over the randomness in outcomes of all edges except $(i,t)$. 
\section{Decomposable Probabilities ($\mb{p_{it}=p_i\cdot p_t}$)}\label{sec:decomp}

Our first algorithm for online matching with stochastic rewards computes a randomly perturbed per-unit reward for every resource and makes each match with the goal of maximizing the expected perturbed reward from the match. 

\begin{algorithm}[H]
	\SetAlgoNoLine
	$S = I$, $g(t)=e^{t-1}$\;
	For every $i\in I$ generate i.i.d.\ r.v.\ $y_i\in U[0,1]$\;
	\For{\text{every new arrival } $t$}{
		$i^*=\underset{i\mid (i,t)\in E, i\in S}{\arg\max} p_{it}r_i (1-g(y_i))$\;
		\textbf{offer} $i^*$ to $t$\;
		\textbf{If} {t accepts $i^*$} \textbf{update} $S=S\backslash\{i^*\}$\;			}

	\caption{Generalized Perturbed-Greedy}
	\label{rank}
\end{algorithm}
Algorithm \ref{rank} matches greedily w.r.t.\ expected \emph{reduced} rewards $p_{it} r_i(1-g(y_i))$. We refer to random variable $y_i$ as the \emph{seed} of resource $i$. The choice of scaling function $g(\cdot)$ naturally influences the performance of the algorithm. We let the analysis dictate this choice. As it turns out, $g(x)=e^{x-1}$ achieves the best possible performance guarantee of $(1-1/e)$.

 When $p_{it}=1\,\; \forall (i,t)\in E$, Algorithm \ref{rank} is exactly the Perturbed-Greedy algorithm of \cite{goel}. \cite{devanur} show that Perturbed-Greedy is $(1-1/e)$ competitive for deterministic online matching using the primal-dual approach. To analyze Algorithm \ref{rank} for stochastic rewards, we first explore a similar primal-dual analysis with the expectation LP benchmark. Since Algorithm \ref{rank} offers a natural extension of Perturbed-Greedy, at the outset it might even seem that the analysis in \cite{devanur} may generalize directly.  
 In the next section, we demonstrate the challenge with this approach through a simple example. Subsequently, we introduce a stronger benchmark, called the path based program (PBP), that we use to show a tight guarantee for Algorithm \ref{rank} against offline. Finally, in Section \ref{sec:obs} we discuss (through an example) the challenges with extending the analysis to non-decomposable probabilities.
\subsection{Challenges with using Expectation LP Benchmark}\label{sec:badLP}

Let \opt\ denote the expected reward of offline. In absence of the stochastic component (all probabilities one), the offline and LP converge and we have $\opt=OPT(LP)$. In general, it is known that there are simple instances with a gap of up to $(1-1/e)$ between $\opt$ and $OPT(LP)$ i.e., $OPT(LP)$ can be as large as $(1-1/e)^{-1}\opt$.\footnote{Consider one unit of a single resource with unit reward and $m$ identical arrivals that have an edge to the resource. Suppose that all edge probabilities are identically equal to $1/m$. The best that any offline algorithm can do is to match each arrival as long as the resource is still available. While this has expected reward $(1-1/e)$ for $m\to +\infty$, the optimal value of LP equals 1 \citep{brubach2,deb}.}
However, the gap between $OPT(LP)$ and $\opt$ does not preclude the possibility of showing a strong competitive ratio guarantee against $OPT(LP)$. After all, a $1/2$ competitive result is known against the expectation LP benchmark in even more general settings \citep{negin}. 

The LP based primal-dual framework of \cite{buchbind} and \cite{devanur} is, arguably, the most versatile and general technique for proving guarantees for online matching and related problems. 
For the stochastic rewards problem, there are fundamental obstacles to using this technique with the expectation LP as a benchmark, and previous work explores novel approaches instead \citep{deb,mehta}. Let us understand one such hurdle from the context of the framework in \cite{devanur}. We start with the dual of the expectation LP,
\begin{eqnarray}
	&\min &\sum_{t\in T}\lambda_t +\sum_{i\in I} \theta_{i} \nonumber\\
	&s.t.\ & \lambda_{t} +p_{it}\theta_i \geq p_{it}r_i\quad \forall (i,t)\in E\label{dualmain}\\
	&& \lambda_t, \theta_{i}\geq 0 \quad \forall t\in T,i\in I\nonumber	
\end{eqnarray} 
\noindent \textbf{Primal-dual certificate~\citep{devanur}:} To prove $\alpha$ competitiveness for an online algorithm with expected reward \alg, it suffices to find a set of non-negative values $\lambda_t,\theta_i$ such that,
\begin{enumerate}[(i)]
	\item$\lambda_t+p_{it}\theta_i\geq\alpha p_{it} r_i,\quad \forall (i,t)\in E,$
	\item $ \alg\geq \sum_{t\in T}\lambda_t + \sum_{i\in I} \theta_i,$ 
\end{enumerate} 

Now, consider Algorithm \ref{rank} and let \alg\ denote both the algorithm and its expected reward. Inspired by \cite{devanur}, we consider the following path based setting of dual variables:

\emph{Given a seed $Y=(y_i)_{i\in I}$ (used in Algorithm \ref{rank}) and sample path $\omega$ over the edge outcomes, when the algorithm offers vertex $i$ to arrival $t$, set $\lambda_t^{\omega}(y_i)= r_i(1-g(y_i))\one(i,t)$ and increase the value of  $\theta^{\omega}_i(y_i)$ by $r_ig(y_i)\one(i,t)$. Let all unset variables be 0.} 

It can be shown that the sum $\sum_t \lambda_t^{\omega}(y_i)+\sum_i \theta^{\omega}_i(y_i)$ equals the total reward of the algorithm on every sample path given by $Y,\omega$. Let $E_{\omega,{Y}}[\cdot]$ denote the expectation w.r.t.\ randomness in ${Y}$ as well as $\omega$. Define $\hat{\lambda}_t=E_{\omega,{Y}}[\lambda^{\omega}_t(y_i)]$ and $\hat{\theta}_i=E_{\omega,{Y}}[\theta^{\omega}_i(y_i)]$. Then, to prove $(1-1/e)$ competitiveness of Algorithm \ref{rank}, it suffices to show that 
that $\{\hat{\lambda_t}\}_{t\in T}$ and $\{\hat{\theta}_i\}_{i\in I}$ satisfy constraint (i) with $\alpha=(1-1/e)$ for every 
edge $(i,t)$. When all edge probabilities are identically 1, for $g(x)=e^{x-1}$, this follows directly from Lemma 2.4 in \cite{devanur}. 
Unfortunately, this is not true for non-unit probabilities. 
Consider the simplest possible example where the first arrival $t=1$ has an edge to resource $i$ and no other edges are incident on $i$ in the future. In this case, \alg\ always matches $i$ to $t$ and we have,
\[ \hat{\lambda}_t=p_{it} r_i E_{y_i}[1-g(y_i)], \text{ } \hat{\theta}_i=p_{it}r_i E_{y_i}[g(y_i)]. \]
Thus,
	\[\hat{\lambda}_t+p_{it}\hat{\theta}_i=p_{it}r_iE_{y_i}[1-g(y_i)]+p^2_{it}r_iE_{y_i}[ g(y_i)].\] 
For small $p_{it}$, any term with $p_{it}^2$ can be ignored and therefore, the expectation approaches $p_{it}r_iE_{y_i}[(1-g(y_i))]$. 
Consequently, for the natural trade-off function $g(x)=e^{x-1}$, we have, $\hat{\lambda}_t+\hat{\theta}_i\leq \frac{1}{e}p_{it}r_i$. 

The main problem stems from fact that the LP only imposes capacity constraints in expectation, giving rise to a $p_{it}^2$ term as we see above. One may consider other ways of setting dual values, such as the one in \cite{buchbind} or \cite{negin}, however the resulting analysis faces similar issues. 
Our path based analysis circumvents this problem by imposing path-based dual constraints.
\subsection{Path Based Program (PBP)}\label{sec:pbp}
We start by constructing a tighter upper bound on offline via a path based formulation. Since \opt\ is deterministic (Lemma \ref{determ}), we consider binary decision variables $x^{\omega}_{it}\in\{0,1\}\,\; \forall (i,t)\in E, \omega\in \Omega$. On a given sample path $\omega$, let $x^{\omega}_{it}$ equal 1 if offline matches $i$ to $t$, and 0 otherwise. Then, the following must be satisfied for every $\omega$,
\begin{eqnarray}
&&   \sum_{t\mid (i,t)\in E} x^\omega_{it}\one(i,t)\leq 1\quad \forall i\in I  \label{cap}\\
&&\sum_{i\mid (i,t)\in E} x^\omega_{it}\leq 1 \quad \forall t\in T  \label{match}
\end{eqnarray}
Constraints~\eqref{cap} capture the fact that any resource $i$ is used at most once on $\omega$. This is in contrast to the LP earlier, where this condition was imposed only in expectation over all sample paths. Similarly, constraints \eqref{match} capture that $t$ is matched to at most one vertex on every sample path. {\color{black} Recall that, offline 
does not know the outcome of an edge before it is chosen to be in the matching.  
Thus, the matching decision for arrival $t$ 
must not depend on the value of $\onee(i,t)$ but may depend on the value of $\onee(j,t')$ for $t'\neq t$. 
We capture this structure with the constraints,}
\begin{equation}
x^{\omega}_{it}=x^{\omega'}_{it}\quad \forall\, \omega,\omega'\in\Omega \text{ such that $ \omitt=\omitt'$.}\label{indep}
\end{equation}
Using \eqref{cap}, \eqref{match} and \eqref{indep}, we now formulate a path based LP relaxation with the objective of maximizing the total expected reward.  
\begin{eqnarray*}
	&\text{PBP}: \max_{} &E_{ {\omega}}\big[\sum_{(i,t)\in E} r_ix^{\omega}_{it}\one(i,t)\big]\nonumber \\
	&\text{s.t. } & \sum_{t\mid (i,t)\in E} x^\omega_{          it}\one(i,t)\leq 1\quad \forall i\in I,\omega\in \Omega\\
	&&\sum_{i\mid (i,t)\in E} x^\omega_{it}\leq 1 \quad \forall t\in T,\omega\in \Omega\\
	&&x^{\omega}_{it}=x^{\omega'}_{it}\quad \forall\, (i,t)\in E,  \{\omega,\omega'\in\Omega\mid  \omitt=\omitt'\}\\
	&&  0\leq x^\omega_{it}\leq 1\quad \forall (i,t)\in E,\omega\in \Omega. \nonumber
\end{eqnarray*}

\noindent Since offline must satisfy all constraints in the program, the values $x^{\omega}_{it}$ generated by executing offline over sample paths $\omega\in\Omega$, yield a feasible solution for PBP. Therefore, the value of PBP offers an upper bound on the performance of the offline benchmark (both clairvoyant and fully offline). 
The natural next step in developing a primal-dual argument is to construct a feasible solution for the dual of PBP. The dual now has exponentially many variables and constraints. Instead of dealing with those directly, the following lemma establishes a weak duality result with a simpler set of constraints i.e., finding a feasible solution for these constraints is easier and more natural. 


\begin{lemma}\label{dual}
	Consider non-negative variables $\lambda^\omega_t, \theta^\omega_{i}$ such that the following inequality holds for every edge $(i,t)$, partial sample path $\omitt$, and some constant $\alpha>0$, \[E_{{\omega}}[\lambda^\omega_t + \theta^\omega_{i}\one(i,t)\mid \omitt]\geq \alpha p_{it}r_i.\] Then, for every feasible solution $\{x^\omega_{it}\}$ of PBP, 
	\[\frac{1}{\alpha} E_{{\omega}}\big[\sum_{t\in T} \lambda^\omega_t +\sum_{i\in I} \theta^\omega_{i}\big]\geq \sum_{(i,t)\in E}p_{it}r_i E_{{\omega}}[ x^\omega_{it}]. \]
	Therefore, $ E_{{\omega}}\big[\sum_{t\in T} \lambda^\omega_t +\sum_{i\in I} \theta^\omega_{i}\big]\geq \alpha\, OPT(\text{PBP})$.
\end{lemma}

\begin{proof}{Proof.}
	Fix an arbitrary feasible solution $\{x^\omega_{it}\}$, then for 
	any given sample path $\omega$ we have,
	\[\sum_{t\in T} \lambda^\omega_t \sum_{i\mid (i,t)\in E} x^\omega_{it}+  \sum_{i\in I} \theta^\omega_{i}\sum_{t\mid (i,t)\in E}  x^\omega_{it} \one(i,t)\leq \sum_{t\in T} \lambda^\omega_t + \sum_{i\in I} \theta^\omega_{i},\]
	here we multiplied inequality \eqref{cap} for each $i$ by $\theta^\omega_{i}\geq0$, inequality \eqref{match} for each $t$ by $\lambda^\omega_t\geq0$, and summed over all $i$ and every $t$. 
	We re-write the inequality above as,
	\begin{eqnarray*}
		 \sum_{t\in T} \lambda^\omega_t +\sum_{i\in I} \theta^\omega_{i}\, \geq \sum_{(i,t)\in E} x^\omega_{it} \lambda^\omega_t + \sum_{(i,t)\in E} x^\omega_{it} \theta^\omega_{i}\one(i,t).
	\end{eqnarray*}
	Taking expectation (convex combination of linear constraints) over $\omega$, 
{\color{black}	\begin{eqnarray*}
		E_{{\omega}}\big[\sum_{t\in T} \lambda^\omega_t +\sum_{i\in I} \theta^\omega_{i}\big]&&\geq E_{{\omega}}\big[\sum_{(i,t)\in E} x^\omega_{it}(\lambda^\omega_t+  \theta^\omega_{i}\one(i,t))\big]\\
		&&= \sum_{(i,t)\in E} E_{\omitt}\big[ E_{{\omega}}[x^\omega_{it}(\lambda^\omega_t+ \theta^\omega_{i}\one(i,t))\mid \omitt]\big]\\
		&&= \sum_{(i,t)\in E} E_{\omitt}\big[x^\omega_{it}\, E_{{\omega}}[(\lambda^\omega_t+ \theta^\omega_{i}\one(i,t))\mid \omitt]\big]\\
		&&\geq \alpha p_{it}r_i\sum_{(i,t)\in E}E_{{\omega}}[ x^\omega_{it}],
	\end{eqnarray*}}
{\color{black}	here the first equation follows from the tower rule of expectation. The second equation follows from the fact that conditioned on $\omitt$, $x^\omega_{it}$ is deterministic (condition \eqref{indep}).} The final inequality follows from, $E_{{\omega}}[\lambda^\omega_t + \theta^\omega_{i}\one(i,t)\mid \omitt]\geq \alpha p_{it}r_i$. To finish the proof consider the objective of PBP.
	\begin{eqnarray*}
		OPT(\text{PBP})=	E_{{\omega}}\big[\sum_{(i,t)\in E} r_ix^{\omega}_{it}\one(i,t)\big]&&=\sum_{(i,t)\in E} E_{\omitt}\big[ E_{{\omega}}\big[r_ix^{\omega}_{it}\one(i,t)\mid\omitt\big]\big]\\
		&&=\sum_{(i,t)\in E}E_{\omitt}\big[ r_ix^{\omega}_{it}E_{{\omega}}[\one(i,t)\mid\omitt]\big]\\
		&&=\sum_{(i,t)\in E}E_{\omitt}\big[ p_{it} r_ix^{\omega}_{it}\big]=\sum_{(i,t)\in E} p_{it} r_iE_{{\omega}}[ x^{\omega}_{it}],
	\end{eqnarray*}
	here the first equation follows form the tower rule of expectation, the second from condition \eqref{indep}, the third equation follows from the independence of $\onee(i,t)$ from $\omitt$ and the final equation also from condition \eqref{indep}.

\hfill\Halmos
\end{proof}
\subsection{Path Based Dual Fitting}\label{sec:pbdf}
Let \alg\ denote Algorithm \ref{rank} as well as its total expected reward. Armed with Lemma \ref{dual}, we now focus on finding `dual' values $\lambda^\omega_t\geq 0$ and $\theta^\omega_{it}\geq 0$ such that,
\begin{eqnarray}
E_{{\omega}}[\sum_{t\in T} \lambda^\omega_t +\sum_{i\in I} \theta^\omega_{i}]&&=\alg, 
\label{rw}\\
E_{{\omega}}[\lambda^\omega_t + \theta^\omega_{i}\one(i,t) \mid\omitt]&&\geq \alpha \,\,p_{it}r_i\qquad \forall (i,t)\in E, \omitt\in \Omega.\label{fsb}
\end{eqnarray} 
If we can find dual values that satisfy the system given by these linear constraints, then using Lemma \ref{dual} we that \alg\ is  $\alpha$ competitive against offline. We start by defining a candidate solution for the system based on the decisions of Algorithm \ref{rank}, followed by a proof of feasibility. This process is called dual fitting \citep{survey}.

Given the fact that we have external randomness in \alg\ due to seeds $\mb{ Y}=({y}_i)_{i\in I}$, we first define variables $\lambda^{\omega,Y}_t$, $\theta^{\omega,Y}_i$ and then set $\lambda^{\omega}_t=E_{Y}[\lambda^{\omega,Y}_t]$ and $\theta^{\omega}_i=E_{Y}[\theta^{\omega,Y}_i]$. To that end, we couple the stochastic rewards in \alg\ and PBP as follows.
\smallskip

\noindent \textbf{Coupling between \alg\ and PBP:}	Given sample path $\omega\in \Omega$ we consider concurrent execution of \alg\ and PBP with match successes/failures given by $\omega$ i.e., both \alg\ and PBP see the same values $\one(i,t),\,\forall (i,t)\in E$. 

Now, consider a path based generalization of the dual setting in \cite{devanur}.  
Initialize all dual variables to 0 and for any match $(i,t)$ in \alg\ set, 
\begin{equation}
\lambda^{\omega,Y}_t=r_i(1-g(y_i))\one(i,t) \text{ and increment } \theta^{\omega,Y}_i \text{ by } r_ig(y_i)\one(i,t).\label{dualset}
\end{equation}
Clearly, $\lambda^{\omega,Y}_t$ is set uniquely since \alg\ offers at most one resource $i$ to arrival $t$. Also, $\theta^{\omega,Y}_i$ takes a non-zero value only if it is also accepted by some $t$ and if this occurs, $\theta^{\omega,Y}_i$ is never re-set. 
We have the following as an immediate consequence of this setting of dual values.
\begin{lemma}\label{dualreward}
	For the candidate solution given by \eqref{dualset}, equality \eqref{rw} is satisfied. 
\end{lemma}
\begin{proof}{Proof.}
	Fix sample path $\omega$ and seed $Y$. Let $M^{\omega,Y}$ denote the matching output by \alg\ on the resulting sample path. By definition of the dual values in \eqref{dualset}, we have that $\sum_t \lambda^{\omega,Y}_t+\sum_i \theta^{\omega,Y}_i$ is exactly $\sum_{(i,t)\in M^{\omega,Y}} \one(i,t) r_i$. Taking expectation gives the desired.
\hfill\Halmos 
\end{proof}
%

It remains to show \eqref{fsb}. Before we proceed, let us reflect back on the issue highlighted in Section \ref{sec:badLP} and see how the path based approach resolves it. Consider a simple instance where the first arrival, $t=1$, has an edge to resource $i$ and no other arrival is adjacent to $i$. \alg\ matches $i$ to $t$ and we have,
\begin{eqnarray*}
	E_{{\omega},Y}[\lambda^{\omega,Y}_t+\theta^{\omega,Y}_i\one(i,t)]&&=r_i  E_{{\omega},Y}\Big[(1-g(y_i))\one(i,t)+g(y_i)\big(\one(i,t)\big)^2\Big],	\\
	&&=r_i  E_{{\omega},Y}[\one(i,t)]=p_{it}r_i.
\end{eqnarray*} 
Clearly, we do not face the issue highlighted in Section \ref{sec:badLP}. 
In the path based approach, we are \emph{able to correlate random variables} $\theta^{\omega,Y}_i$ and $\one(i,t)$, and \emph{boost the expectation of the product} $\theta^{\omega,Y}_i\one(i,t)$. In contrast, the dual feasibility condition coming out of the expectation LP has the term $p_{it} \theta_i=E_{{\omega}}[\one(i,t)] E_{{\omega},\mb{Y}}[\theta^{\omega,Y}_i]$, which does not allow for the possibility of correlating the two random variables. 

We now show that inequalities \eqref{fsb} are satisfied with $\alpha=(1-1/e)$, as long as the probabilities are decomposable. In fact, we show a stronger statement as described in the following lemma. 
\begin{lemma}\label{prop1}
		Consider an edge $(i,t)\in E$. Suppose that for the dual fitting given by \eqref{dualset}, for every sample path $\omitt$ and values $y_j$ for all $j\neq i$, denoted $Y_{-i}$, we have,
	\begin{equation}\label{condifsb}
		E_{y_i}\left[E_{\omega}\left[\lambda^{\omega,Y}_t + \theta^{\omega,Y}_i \one(i,t)\mid \omitt\right]\, \big|\, Y_{-i}\right]\geq (1-1/e)p_{it}r_i.
	\end{equation}
Then, inequality \eqref{fsb} is satisfied for edge $(i,t)$ with $\alpha=(1-1/e)$.
	\end{lemma}
\begin{proof}{Proof.} Taking expectation over $Y_{-i}$ and using the tower rule of expectation we have, $E_{\omega}\left[E_{Y}\left[\lambda^{\omega,Y}_t + \theta^{\omega,Y}_i \one(i,t)\right]\mid \omitt\right]\geq (1-1/e)p_{it}r_i,\, \forall \omitt.$
	Taking expectation over $\omitt$ completes the proof.
	\hfill	\Halmos
\end{proof}

	For the following discussion, fix an arbitrary edge $(i,t)$ {\color{black} (not necessarily chosen in the matching)}, sample path $\omitt$, and random seed $Y_{-i}$ for every $j\neq i$. For brevity, we omit $\omitt$ and $Y_{-i}$ from notation and let, 
	\[E_{y_i}\left[E_{\omega}\left[\lambda^{\omega,Y}_i\right]\right]:=E_{{y}_i}\left[E_{\omega}\left[\lambda^{\omega,Y}_t\mid \omitt\right]\mid Y_{-i}\right].\]
	Similarly,
	\[E_{{y}_i}\left[E_{\omega}\left[\theta^{\omega,Y}_i \one(i,t)\right]\right]:=E_{{y}_i}\left[E_{\omega}\left[\theta^{\omega,Y}_i \one(i,t)\mid \omitt\right]\mid Y_{-i}\right].\]
To show \eqref{condifsb}, we separately lower bound $E_{y_i}\left[E_{\omega}\left[\lambda^{\omega,Y}_i\right]\right]$ (Lemma \ref{lambda}) and $E_{{y}_i}\left[E_{\omega}\left[\theta^{\omega,Y}_i \one(i,t)\right]\right]$ (Lemma \ref{theta}). Combining these lower bounds gives us the desired. The lower bound on $E_{y_i}\left[E_{\omega}\left[\lambda^{\omega,Y}_i\right]\right]$ is similar to Lemma 2.3 {\color{black} (Monotonicity Lemma) in \cite{devanur}. At a high level, the lower bound on $E_{{y}_i}\left[E_{\omega}\left[\theta^{\omega,Y}_i \one(i,t)\right]\right]$ follows from a ``dominance" property similar to Lemma 2.2 in \cite{devanur}. However, our analysis has an interesting subtlety (absent in the classical case) due to the fact that matches can be unsuccessful.} This is also where we need the decomposability assumption for edge probabilities. We discuss this further in the proof of Lemma \ref{theta}. 

 With $\omitt$ and $Y_{-i}$ fixed, let $M_{-i}$ denote the matching generated by \alg\ up to (and including) arrival $t$, when it is executed with the reduced set of vertices $I\backslash{i}$. 
 Unlike the case where all edge probabilities are 1, here $M_{-i}$ may have one offline vertex matched multiple times though only one match could have actually succeeded. Let $I^{t}_{-i}$ denote the set of available neighbors of $t$, in this execution of \alg\ without vertex $i$. Define $y^c_i$ such that, \[p_{it}r_i(1-g(y^c_i))=\max_{j\in I^{t}_{-i}} p_{jt}(1-g(y_j)).\]   
Set $y^c_i=1$ if the set $I^t_{-i}$ is empty, and $y^c_i=0$ if  no such value exists. 
Due to the monotonicity of $g(t)=e^{t-1}$, we have a unique value of $y^c_i$. 

\begin{lemma}\label{lambda}
	Given edge $(i,t)$, sample path $\omitt$, seed $Y_{-i}$, we have  $E_{\omega}[\lambda^{\omega,Y}_t]\geq p_{it}r_i(1-g(y^c_i))$ for every $y_i\in[0,1]$. Thus, $E_{y_i}\left[E_{\omega}\left[\lambda^{\omega,Y}_i\right]\right]\geq p_{it}r_i(1-g(y^c_i))$.
\end{lemma}
\begin{proof}{Proof.}
Given edge $(i,t)$, sample path $\omitt$, and random seed $Y_{-i}$, let $I^{t}_{y_i}$ denote the set of available neighbors at $t$ when \alg\ is executed with the full vertex set and seed value $y_i$.
	Observe that if, 
	\[\max_{j\in I^{t}_{y_i}} p_{jt}(1-g(y_j))\geq p_{it}(1-g(y^c_i))\quad \forall y_i\in[0,1],\] 
	then we have the desired. To show this inequality, it suffices to prove that $I^t_{-i}\subseteq I^t_{y_i}\, \forall y_i\in[0,1]$. Let $M_{y_i}$ denote the matching generated by \alg\ up to (and including) $t$ with seed $y_i$. The proof follows from the following cases. 
	 
	 \noindent \textbf{Case I:} In $M_{y_i}$, $i$ is not matched to any arrival prior to $t$. Thus, matchings $M_{-i}$ and $M_{y_i}$ are identical prior to arrival $t$ and 
	 $I^t_{y_i}=I^t_{-i}\cup \{i\}$. 

\noindent \textbf{Case II:} In $M_{y_i}$, $i$ is matched prior to $t$. Let $t'$ denote the first arrival matched to $i$ in $M_{y_i}$. For the sake of contradiction, let $\tau\leq t$ be the first arrival such that there exists a resource $j\in I^\tau_{-i}\backslash I^\tau_{y_i}$. Clearly,  $\tau> t'$ since 
$M_{-i}$ and $M_{y_i}$ are identical prior to $t'$. Now, let $t_j<\tau$ be the arrival successfully matched to $j$ in $M_{y_i}$. So, we have the arrivals in the following order $t'\leq t_j<\tau\leq t$. Since $j\in I^{\tau}_{-i}$, we have, $j\in I^{t_{j}}_{-i}$. 
Further, $I^{t_j}_{-i}\subseteq I^{t_j}_{y_i}$. Thus, $j\in I^{t_j}_{y_i}$ and by definition of \alg, $j$ must be (successfully) matched to $t_j$ in $M_{-i}$, contradiction. Thus, $I^\tau_{-i}\subseteq  I^\tau_{y_i}$ for every $\tau\leq t$. 

\hfill\Halmos 

\end{proof}

\begin{lemma}\label{theta}
	Given edge $(i,t)$, sample path $\omitt$, seed $Y_{-i}$ and decomposable edge probabilities $p_{it}=p_ip_t\, \forall (i,t)\in E$, we have $E_{{y}_i}\left[E_{\omega}\left[\theta^{\omega,Y}_i \one(i,t)\right]\right]\geq p_{it}r_i\int_{0}^{y^c_i} g(x) dx$.
\end{lemma}
\begin{proof}{Proof.}
	Given edge $(i,t)$, sample path $\omitt$, and random seed $Y_{-i}$, consider a value $y_i\in [0,y^c_i)$. There are two possibilities, either $i$ is successfully matched before $t$ and unavailable at $t$, or $i$ is available when $t$ arrives. We will show that in the latter scenario \alg\ matches $i$ to $t$ but first, let us see how this proves the lemma. When $i$ is unavailable at $t$, $\theta^{\omega,Y}_i=r_ig(y_i)$ since $i$ was successfully matched to an arrival preceding $t$. In case $i$ is available and thus matched to $t$ (by assumption), $\theta^{\omega,Y}_i=r_ig(y_i)\one(i,t)$. Therefore, in both cases $E_{\omega}[\theta^{\omega,Y}_i\one(i,t)]=p_{it}r_ig(y_i)$. Since this holds for all values of $y_i< y^c_i$ we have the desired. 
	
	To finish the proof we need to show that, for $y_i< y^c_i$, $i$ is matched to $t$ if available. In the classical/deterministic case this follows directly from $r_i(1-g(y_i))\geq r_i(1-g(y^c_i))$. 
	In our case, we still have $p_{it}r_i(1-g(y_i))\geq p_{it}r_i(1-g(y^c_i))$. However, for some small enough value of $y_i$, say $y_0$, $i$ may be unsuccessfully matched to some arrival $t'$ preceding $t$. This could free up a resource $j$($\neq i$) that was successfully matched to $t'$ for $y_i>y_0$. So for $y_i\leq y_0$, the freed up resource $j$ may be matched to $t$ instead of $i$, even if $i$ was available at $t$. In such a scenario it is not a priori clear if we still have the desired lower bound on $\theta^{Y}_i$. Observe that if $i$ is never matched/offered to any arrival preceding $t$ 
	then such an issue does not arise and the claim follows as in the deterministic case. So the interesting case is when $i$ is matched (unsuccessfully) prior to $t$ for some value $y_i<y_i^c$ and consequently, $t$ is matched to some $j$ ($\neq i$) such that, $E[\lambda^{Y}_t]=p_{jt}r_j(1-g(y_j))>p_{it}r_i(1-g(y^c_i))$. In this case, consider the graph given by the set difference between the matching $M_{y_i}$ generated by \alg\ with seed value $y_i$ and the matching $M_{-i}$ generated by \alg\ when resource $i$ is excluded. On this difference graph, there exists a unique alternating path that includes both $i$ and $t$. 
	Since \alg\ matches every arrival to the resource with maximum expected reduced reward, moving along the alternating path starting with the edge incident on $i$ and applying the decomposition of probabilities, we have for every edge $(i',t')$ on the path, $p_{i}r_i(1-g(y_i))\geq p_{i'}r_{i'}(1-g(y_{i'}))$ (though it may be that, $p_{it}r_i(1-g(y_i))< p_{i't'}r_{i'}(1-g(y_{i'}))$ ). In particular, $p_{i}r_i(1-g(y_i))\geq p_{j}r_{j}(1-g(y_{j}))$ and thus, $p_{it}r_i(1-g(y_i))\geq p_{jt}r_{j}(1-g(y_{j}))$. Therefore, if $i$ is available on arrival of $t$, the algorithm matches $i$ to $t$ (since ties occur w.p. 0). 
\hfill\Halmos 
\end{proof}

\begin{proof}{Proof of Theorem \ref{maint}.} The proof follows by combining Lemmas \ref{dual}, \ref{dualreward}, \ref{lambda} and  \ref{theta}, and using the fact that $1-g(y)+\int^y_o g(x) dx=(1-1/e)\, \forall x\in[0,1]$ when $g(x)=e^{x-1}$.
\hfill\Halmos
 \end{proof}
\subsection{Obstacle for Non-decomposable Probabilities}\label{sec:obs}

To understand the effect of decomposable probabilities in Algorithm \ref{rank}, note that the values $p_{i}r_i (1-g(y_i))$  induce an ordering over the vertices $i\in I$ that is independent of arrivals. Therefore, given any two resources $i,j$, and an arrival $t$ with an edge to both, Perturbed-Greedy prefers to offer $i$ over $j$ whenever $p_ir_i (1-g(y_i))>p_jr_j (1-g(y_j))$. This fixed ordering/ranking is a key component in the proof of Lemma \ref{theta}. 
For general probabilities, Algorithm \ref{rank} may prefer $j$ over $i$ for arrival $t$ but vice versa for another arrival $t'$. 

Let us see an example with non-decomposable probabilities where Lemma \ref{theta} fails. Consider a 3x3 bipartite graph with arrival $t_1$ arriving first, followed by $t_2$, and finally $t_3$. Consider resources $i,j,k$ with resource $i$ adjacent to all three arrivals, resource $j$ adjacent to arrivals $t_2$ and $t_3$ and finally, resource $k$ adjacent solely to arrival $t_1$. 
Let the rewards $r_i,r_j,$ and $r_k$ be identical and equal 1. Consider probabilities $p_{it_1}=p_{kt_1}$, $p_{it_2}=p_{jt_2}$ and $p_{it_3}< p_{jt_3}$. Note that such probabilities are not decomposable. Let us focus on the dual feasibility of constraint corresponding to edge $(i,t_3)$. So we fix $y_j$, $y_k$ with $y_j>y_k$ and consider the following sample path before $t_3$ arrives, $\one(k,t_1)=0,\one(i,t_1)=1$ and $\one(i,t_2)=0$, $\one(j,t_2)=1$. Observe that $y^c_i=1$ and consider $E[\lambda^Y_{t_3}]$ and $E[\theta^Y_i \onee(i,t_3)]$ as $y_i$ varies. 

For $y_i\geq y_j$, $k$ is offered to $t_1$, $j$ is offered to and accepted by $t_2$, and $i$ is offered to $t_3$. Therefore, $E[\theta^Y_i \onee(i,t_3)]=p_{it_3}r_ig(y_i)$ and $E[\lambda^Y_{t_3}]=p_{it_3}r_i(1-g(y_i))$. 

For $y_k<y_i<y_j$, $i$ is offered to but not accepted by $t_2$ and suppose $p_{it_3}$ is sufficiently smaller than $p_{jt_3}$ so $j$ is offered to $t_3$. Thus, $E[\lambda^Y_{t_3}]=p_{jt}r_j(1-g(y_j))$, but $\theta^Y_i=0$ for $y_k<y_i<y_j$. 

Finally, for $y_i\leq y_k$, $i$ is offered to and accepted by $t_1$, $j$ is offered to and accepted by $t_2$ and $t_3$ is unmatched with $\lambda^Y_{t_3}=0=p_{it_3}r_i(1-g(y^c_i))$ and $E[\theta^Y_i \onee(i,t_3)]=p_{it_3}r_ig(y_i)$. 

Combining the pieces we have, $E_{{y}_i}[E[\lambda^Y_{t_3}+\theta^Y_i \onee(i,t_3)]]=p_{it_3}r_i(1-y_j) + p_{jt_3}{r_j}(1-g(y_j))(y_j-y_k)+p_{it_3}r_{i}\int_0^{y_k}g(x) dx$. Clearly, Lemma \ref{theta} does not hold and the previous expectation can be $<0.5 p_{it_3}r_i$. More concretely, let $y_j=1-\epsilon$ and $p_{jt_3},p_{it_3}$ be such that $p_{jt_3}(1-g(y_j))=p_{it_3}(1-g(y_k))$. Then for $g(x)=e^{x-1}$,
\begin{eqnarray*}
	E_{{y}_i}[E[\lambda^Y_{t_3}+\theta^Y_i \onee(i,t_3)]]&&=c\epsilon+ p_{it_3}{r_i}(1-g(y_k))(1-y_k)+p_{it_3}r_{i}\int_0^{y_k}g(x) dx\\
	&&\geq  c\epsilon+ p_{it_3}{r_i} \min_y \Big([1-g(y)](1-y)+\int_0^y g(x) dx\Big)\\
	&&= c\epsilon+ 0.44 p_{it_3}{r_i}, 
\end{eqnarray*}
where $c$ is some non-negative real number and the inequality is tight for $y_k=0.5571$.
This shows that the existing analysis does not lead to a $(1-1/e)$ guarantee for general edge probabilities. Note that using our analysis for $g(x)=1/2$ (which corresponds to the greedy algorithm) we get a guarantee of $1/2$ for general probabilities.

\section{Vanishing Probabilities ($\mb{p_{it}\to0}$)}\label{sec:vanish}
In this section, we consider the case of vanishingly small heterogeneous probabilities. As demonstrated in the previous section, the path based primal-dual analysis does not yield a guarantee better than 0.5 in this setting. We overcome these challenges by introducing a new algorithm and a novel LP free analysis. Our algorithm is deterministic, and this feature eases some of the analytical challenges. It is worth noting that for arbitrary edge probabilities (including 
the classic online bipartite matching problem), 
no deterministic online algorithm can have a guarantee strictly better than 0.5 \citep{kvv}.  

\begin{algorithm}[H]
	\SetAlgoNoLine
	\textbf{Input:} Scaling function $g(.): [0,\infty)\to[0,1]$, offline vertex set $I$\;
	$S = I$\;
	For every $i\in I$, $l_i=0$\;
	\For{\text{every new arrival } $t$}{
		$i^*=\underset{i\mid (i,t)\in E, i\in S}{\arg\max} p_{it}r_ig(l_i)$\;
		\textbf{offer} $i^*$ to $t$\;
		\textbf{If} {t accepts $i^*$} \textbf{update} $S=S\backslash\{i^*\}$\;			
		\textbf{Else} {$l_{i^*}=l_{i^*}+p_{{i^*}t}$};\ }
	\caption{Generalized Fully-Adaptive}
	\label{fullada}
\end{algorithm}
\smallskip

Observe that Algorithm \ref{fullada} matches greedily w.r.t.\ expected \emph{reduced} rewards $p_{it}r_ig(l_i)$. Unlike Algorithm \ref{rank}, the reduced rewards are now deterministic and change over time. Our algorithm generalizes the Fully-Adaptive algorithm proposed, but not analyzed, in \cite{deb} and \cite{mehta}.   
The choice of scaling function $g(\cdot)$ naturally influences the performance of the algorithm. At a high level, it is desirable to choose a decreasing function so that 
the reduced reward of a resource decreases every time it is matched. The guarantee of 0.596 stated in Theorem \ref{mainh} is obtained with $g(t)=e^{t+1}(\int_{t+1}^{\infty} \frac{e^{-y}}{y} dy)$. We also consider other choices of $g(\cdot)$. In particular, we show that a simple inverse scaling function $g(t)= \frac{1}{\beta t+1}$, is at least 0.588 competitive. For the function $g(t)= e^{-\beta t}$, a natural generalization of the function $e^{-t}$ proposed in \cite{mehta}, we establish 0.581 competitiveness. 

As we saw, our previous path based analysis is insufficient when probabilities do not decompose. In fact, the linear system given by \eqref{rw} and \eqref{fsb} may not be feasible for $\alpha>0.5$. We resolve this difficulty via two new ideas. First, we introduce an alternative viewpoint of the reward process that allows a natural characterization of the probability of an algorithm successfully matching a given resource. This viewpoint generalizes an insight from \cite{deb}, and is introduced in Section \ref{sec:vanish1} and Section \ref{sec:trunc}. 
Then, we convert the problem of finding a competitive ratio guarantee for Algorithm \ref{fullada}, to proving the feasibility of linear system. We introduce a new linear system that is a more relaxed (and easier to satisfy) than LP duality based linear systems. For instance, recall that the system given by \eqref{rw} and \eqref{fsb} has a constraint for every edge in the graph (local constraints). In contrast, we now impose a single constraint for every resource (global constraints). The constraints are path-based and inspired by duality but do not necessarily correspond to the dual of a program. This makes the analysis LP \emph{free} and, generally speaking, robust to any issues related to tightness of a mathematical program. 
 We introduce the LP free system in Section \ref{sec:vanish2}. 
 Finally, in Section \ref{sec:fsb} we utilize these ingredients to prove Theorem \ref{mainh}. 


\subsection{An Alternative Viewpoint}\label{sec:vanish1}
Consider a \emph{deterministic} algorithm $\mathcal{A}$ (online or offline) for the stochastic rewards problem. Let $\mathcal{A}$ be \emph{non-anticipative} i.e., the outcome (success/failure) of every match is revealed to $\mathcal{A}$ only after the match is made and any vertex $t\in T$ can be matched at most once. Observe that, both Algorithm \ref{fullada} and the offline benchmark are deterministic. 

We interpret outcomes of matches as a sequential input to $\mathcal{A}$ i.e., every time $\mathcal{A}$ makes a match, an independent Bernoulli sample denoting success/failure of that match is provided to $\mathcal{A}$. \cite{deb} considered the case of identical edge probabilities $p$ and observed that instead of letting the success of each match be determined by an i.i.d.\ random variable, one could consider the following equivalent process. For each resource $i$, independently generate a starting budget $b_i$ from the distribution,
\begin{equation}\label{identical}
	\mathbb{P}\big(b_i=k p\big)=p (1-p)^{k-1}, \text{  for  } k\in\{1,2,\cdots\}.
\end{equation}
In executing $\mathcal{A}$, every match to $i$ reduces the remaining budget by $p$ and a match succeeds when the remaining budget reaches zero. So if the starting budget $b_i=k p$, the $k$-th match made to $i$ would succeed, and none before. Budgets $(b_i)_{i\in I}$ are sampled in advance but $\mathcal{A}$ is unaware of these values. \cite{deb} showed that budgets $(b_i)_{i\in I}$ define an alternative characterization of stochastic rewards i.e., given a matching, the probability that $\mathcal{A}$ outputs this matching is identical in both methods of generating random match outcomes. 

For $p\to 0$, the distribution of budgets, \eqref{identical}, is well approximated by $Exp(1)$, which denotes the exponential distribution with cumulative distribution $F(x)=1-e^{-x}$ for $x\geq 0$. To analyze Algorithm \ref{fullada}, we generalize this viewpoint (and approximation) 
to arbitrary probabilities. 

\smallskip

\noindent \textbf{Generalized Alternative Viewpoint}: Let $U[0,1]$ denote the uniform distribution on interval $[0,1]$. 
We independently sample random variables $u_i\in U[0,1]$ for $i\in I$. These values are not known to $\mcala$ but for every match we provide the outcome to $\mcala$ based on these values. In other words, $\mb{u}=(u_i)_{i\in I}$ gives a sample path over random  match outcomes. To describe how $\mb{u}$ translates to match outcomes, suppose that resource $i$ has been unsuccessfully matched to $s\geq0$ arrivals in the ordered set $S=\{t_1,\cdots,t_{s}\}$ and $\mcala$ matches a new arrival $t_{s+1}$ to $i$. The match succeeds if $1-\prod_{j\leq s+1} (1-p_{it_j})\geq u_i$, and fails otherwise. Observe that $1-\prod_{j\leq s+1} (1-p_{it_j})$ is the probability that at least one of the matches to $S\cup\{t_{s+1}\}$ succeeds.  Thus, $u_i$ is a threshold such that \emph{$i$ is successfully matched to the first arrival where the probability of success exceeds} $u_i$.  

\begin{lemma}\label{same}
Consider a non-anticipative and deterministic algorithm $\mcala$. Let $\mcala(\omega)$ denote the matching output by $\mcala$ on sample path $\omega\in \Omega$. 
Similarly, let $\mcala(\ub)$ denote the matching output by $\mcala$ on sample path $\ub\in[0,1]^{n}$ in the threshold viewpoint. For every matching $\mathcal{M}$, we have $\mathbb{P}(\mcala(\omega)=\mathcal{M})=\mathbb{P}(\mcala(\ub)=\mathcal{M})$.
	\end{lemma}

See Appendix \ref{appx:vanish1} for a proof.
 To connect thresholds $\ub$ to the random budget viewpoint introduced in \cite{deb}, we start by exploring how the matching generated by $\mcala$ changes as we vary a single threshold value while keeping all other values fixed. To this end, consider a resource $i\in I$ and fix threshold values, $\umi\in[0,1]^{n-1}$, for all resources except $i$. 
 Let $\mcala^{\umi}(1)$ denote the matching output by $\mcala$ when the threshold for $i$ is 1 and threshold for every other resource is given by $\umi$. Let
\[\mcala^{\umi}_i(1)=\{t_1,\cdots,t_{m_{\mcala}^{\umi}}\},\]
 denote the ordered set of $m_{\mcala}^{\umi}$ arrivals matched to $i$ in matching $\mcala^{\umi}(1)$. 
 \begin{lemma}\label{basicu}
 	Consider $i\in I$, $\umi\in[0,1]^{n-1}$, and a deterministic non-anticipative algorithm $\mcala$. Let $\mcala^{\umi}_i(u_i)$ denote the ordered set of arrivals matched to $i$ in the matching $\mcala(\ui,\umi)$. 
 	\begin{enumerate}[(i)]
 		\item For every $k\in[m^{\umi}_{\mcala}]$, $i$ is matched to $t_k\in \mcala^{\umi}_i(1)$ if and only if $u_i>1-\prod_{j\leq k-1} (1-p_{it_j}) $. Thus, conditioned on $\umi$, $i$ is matched to $t_k$ with probability $\prod_{j\leq k-1} (1-p_{it_j})$.
 		\item  For every $u_i\in[0,1]$, there exist some $k\in[m^{\umi}_{\mcala}]$ such that $\mcala^{\umi}_i(u_i)=\{t_1,\cdots,t_k\}$. Further,  $\mcala^{\umi}_i(u_i)\subseteq \mcala^{\umi}_i(u'_i)$ for every $0\leq u_i\leq u'_i\leq 1$. 
 	\end{enumerate}
 	\end{lemma}
 A proof is provided in Appendix \ref{appx:vanish1}. Next, we connect thresholds $\ub$ to budgets $\{b_i\}_{i\in I}$.

%
 \noindent \textbf{From Thresholds to Budgets:} Conditioned on thresholds $\umi$, consider the set $\mcala^{\umi}_i(1)$ of $m^{\umi}_{\mcala}$ arrivals that are sequentially matched to $i$ in $\mcala$. Let 
 \[p(k)=\sum_{j\leq k}p_{it_j},\forall k\in[m^{\umi}_{\mcala}].\] 
 In the random budget viewpoint, {\color{black} for any given resource $i$,  conditioned on thresholds $\umi$}, we independently draw a starting budget $b_i$ from a suitable distribution (defined later). Every time resource $i$ is matched, we reduce the remaining budget by the probability of the corresponding edge. The match is successful if the remaining budget is $\leq 0$.  In other words, \emph{the first match where the sum of the probabilities of all edges matched to $i$ is greater than or equal to the starting budget $b_i$, is successful, and all prior matches are failures}. 
 
{\color{black} Conditioned on thresholds $\umi$}, let $F^{\umi}_{\mcala}(x)$ denote the probability distribution of budget $b_i$. Define \[p_i=\max_{t\mid (i,t)\in E} p_{it}\, \text{ and }\, p=\max_{i\in I} p_i.\] The support of $F^{\umi}_{\mcala}(x)$ must include the set $\{p(1),\cdots,p(m^{\umi}_{\mcala})\}$. In addition, we also include the set $\{p(m^{\umi}_{\mcala})+p_i, p(m^{\umi}_{\mcala})+2p_i, p(m^{\umi}_{\mcala})+3p_i,\cdots, +\infty\}$ in the support. For simplicity, let $p(\ell)=p(m^{\umi}_{\mcala})+(\ell-m^{\umi}_{\mcala}) p_i$ for every $\ell\geq m^{\umi}_{\mcala}$. 
On 
$\{p(1),\cdots,p(m^{\umi}_{\mcala})\}$, the distribution is (uniquely) defined by probabilities,
 \begin{equation}\label{threshold}
 \mathbb{P}\left(x= p(k) \right)=p_{it_k} \prod_{j=1}^{k-1}(1-p_{it_j})\quad \forall k\in[m^{\umi}_{\mcala}].
\end{equation}
On $\{p(m^{\umi}_{\mcala}+1), p(m^{\umi}_{\mcala}+2), \cdots, +\infty\}$, 
we define the distribution as,
\[\mathbb{P}\left(x= p(m^{\umi}_{\mcala}) + \ell p_i \right)\,=\, p_{i}\,  \left(1-P(x\leq p(m^{\umi}_{\mcala}))\right)\, \prod_{j=1}^{\ell-1}(1-p_{i})\quad \forall \ell\in\{1,2,\cdots,+\infty\}. \]
This part of the distribution 
  can be set differently as long as, $\mathbb{P}\left(b_i>p(m^{\umi}_{\mcala})\right)=\prod_{j=1}^{m^{\umi}_{\mcala}}(1-p_{it_j})$. We made a choice that is convenient for analysis. Let $G^{\umi}_{\mcala}:[0,1]\to [0,+\infty)$ denote the inverse function such that $F^{\umi}_{\mcala}\left(G^{\umi}_{\mcala}(u)\right)=u$. An alternative way to generate starting budget $b_i\sim F^{\umi}_{\mcala}$ ({\color{black}conditioned on $\umi$}) is as follows: Independently sample threshold $u_i\in[0,1]$ and set budget $b_i=G^{\umi}_{\mcala}(u_i)$. 

\begin{corollary}\label{basiceq}
	Consider a non-anticipative deterministic algorithm $\mathcal{A}$, resource $i\in I$, and condition on thresholds $\umi\in[0,1]^{n-1}$. Let $\mcala_i^{b,\umi}(x)$ denote the ordered set of arrivals matched to $i$ in the budget viewpoint with starting budget $x$. Then,
	\begin{enumerate}[(i)]
		\item $\mcala^{b,\umi}_i\left(G^{\umi}_{\mcala}(u)\right)=\mcala^{\umi}_i(u),\, \forall u\in[0,1]$. 
		\item The probability that $i$ is successfully matched in $\mathcal{A}$ equals $\mathbb{P}\left(G^{\umi}_{\mcala}(u)\leq p(m^{\umi}_{\mcala})\right).$
	\end{enumerate}  
\end{corollary}
The corollary is a direct consequence of Lemma \ref{basicu} and the definition of the budget viewpoint.
%
\smallskip

\noindent \textbf{Remarks:} We refer to $p(m^{\umi}_{\mcala})$ as the \emph{effort threshold} of $i$ conditioned on $\umi$. A resource is successfully matched when its effort threshold no less than the budget. This generalizes the budget viewpoint in \cite{deb}. An important difference between the generalized budget viewpoint and the special case of identical probabilities is that the distribution of threshold $b_i$ is a function of $\umi$. For identical probabilities, the distribution is independent of $\umi$. Note that while the distribution depends on $\umi$, the budget is sampled independently for every $i\in I$.

\subsection{Approximation by $\mb{Exp(1)}$}\label{sec:trunc}


Now, we consider the case of small edge probabilities and show that the exponential distribution, $Exp(1)$, approximates $F^{\umi}_{\mcala}$. 
Let 
\[x = (1\pm \delta) y \quad \Rightarrow\quad  (1-\delta) y\leq x\leq (1+\delta) y.\]
{\color{black} Let $O(x)$ denote $c\cdot x$, where $c$ is a constant that is independent of $x$ and the problem parameters $\{p_{it}\}_{(i,t)\in E},\{r_i\}_{i\in I},n$ and $T$. Notice that $x= (1\pm O(\delta)) y \Leftrightarrow y= (1+O(\delta)) x$.} 
\begin{lemma}\label{expapprox}
	Given a non-anticipative deterministic algorithm $\mathcal{A}$, resource $i\in I$, threshold vector $\umi\in[0,1]^{n-1}$ and $p\leq 0.38$, 
	we have, 
	\begin{enumerate}[(i)]
		\item $1-F^{\umi}_{\mcala}(x) = (1\pm O(\sqrt{p_i}))
		\,\, e^{-x}\quad \forall x\in[0,\frac{1}{\sqrt{p_i}}].$
		\item $ F^{\umi}_{\mcala}(x) = (1\pm 
		O(\sqrt{p_i}))\,\, (1-e^{-x})\quad \forall x\geq 0.$
%
	\end{enumerate} 
\end{lemma}

{\color{black} We give a self-contained proof of the lemma in Appendix \ref{appx:trunc}. Note that, it may also be possible to prove this lemma using using Le Cam's theorem \citep{lecam}.} From this lemma, we have that $Exp(1)$ is a good approximation for budget distribution $F^{\umi}_{\mcala}$ when the effort threshold $p(m^{\umi}_{\mcala})\leq 1/\sqrt{p_i}$. This approximation will play an important role in analyzing Algorithm \ref{fullada}. The next lemma gives another useful means of approximating the true distribution of budgets by $Exp(1)$. In the lemma, we use $E_{u_i}[\cdot]$ to denote expectation w.r.t.\ the randomness in threshold $u_i\sim U[0,1]$. A proof of the lemma is provided in Appendix \ref{appx:trunc}. 
\begin{lemma}\label{funcapx}\label{singfunc}
		Consider a deterministic non-anticipative algorithm $\mcala$, resouce $i\in I$, threshold vector $\umi\in [0,1]^{n-1}$, a non-decreasing {\color{black} integrable function $h: \mathbb{R}^+ \to \mathbb{R}^+$ such that,
	\begin{enumerate}[(i)]
		\item $h(0)=0$,
		\item $h(x)=h(\frac{1}{\sqrt{p_i}})\,\, \forall x\geq \frac{1}{\sqrt{p_i}}$.
\end{enumerate}} 
For $p\leq 0.38$, we have,
\[ E_{u_i} [h\left(G^{\umi}_{\mcala} (u_i)\right)] = (1\pm O(\sqrt{p_i}))\, E_{b \sim Exp(1)} [h(b)]. \]
	\end{lemma}

Recall that both Algorithm \ref{fullada} and the offline benchmark are deterministic. 
Therefore, the lemmas in this section and in Section \ref{sec:vanish1}, apply to both Algorithm \ref{fullada} and the offline benchmark.
In the next section, we present a new path-based linear system and show that proving feasibility of the system establishes a competitive ratio guarantee for Algorithm \ref{fullada}. In Section \ref{sec:fsb}, we use the alternative viewpoint of stochastic rewards (and its approximation) to prove feasibility of this system. 
\subsection{LP Free Framework}\label{sec:vanish2}

 
In this section, we define a system of linear inequalities such that the existence of a feasible solution to this system implies a competitive ratio guarantee for Algorithm \ref{fullada}. Our linear system has a single constraint for every offline vertex as opposed to a constraint for every edge. 
The new feasibility constraints are path-based and inspired by duality but do not necessarily correspond to the dual of a program. This makes the analysis LP \emph{free} and generally speaking, robust to any issues related to tightness of a mathematical program. We start by introducing notation. 

Let \opt\ denote the fully offline algorithm as well as its expected total reward. Similarly, let \alg\ denote denote Algorithm \ref{fullada} and its expected reward. Given a resource $i\in I$, and thresholds $\umi\in[0,1]^{n-1}$, let $p(m^{\umi}_{\alg})$ and $p(m^{\umi}_{\opt})$ denote the effort threshold of resource $i$ in \alg\ and \opt\ respectively. We use $E_{\ub}[\cdot]$ to denote expectation w.r.t.\ the randomness in thresholds $\ub$. For every $i\in I$, let $E_{\umi}[\cdot]$ denote expectation w.r.t.\ randomness in $\umi$.

The variables of the system are $\lambda^{\umi}_{it}$ and $\theta^{\umi}_i$, for every $\umi\in[0,1]^{n-1}, i\in I, (i,t)\in E$. 
The system is defined as follows,
			\begin{eqnarray}
		\theta^{\umi}_i+ \sum_{t\mid (i,t)\in E} \lambda^{\umi}_{it} \geq \alpha r_i\, \left(1-e^{-p(m^{\umi}_{\opt})}\right)
			\quad 
			\forall i\in I,\umi\in [0,1]^{n-1}, \label{dual1}\\
			\beta \alg\, \geq \sum_{i\in I}E_{\umi}\left[\theta^{\umi}_i\right] + \sum_{(i,t)\in E} E_{\umi}\left[\lambda^{\umi}_{it}\right],\label{dual2}\\
			\theta^{\umi}_i,\lambda^{\umi}_{it} \geq 0\quad \forall \umi\in [0,1]^{n-1}, i\in I, t\in T, \label{dual3}
			\end{eqnarray}
where $\alpha,\beta \geq0$ are constants that determine the competitive ratio guarantee. Linearity of the system follows from linearity of the expectation operator.  In Appendix \ref{appx:huang}, we compare our LP free framework with the LP based primal-dual approach employed in \cite{huang}. 

The next lemma shows that the feasibility of this linear system implies a competitive ratio guarantee for $\alg$. 

\begin{lemma}\label{final}
	Given a feasible solution 
	to the system defined by \eqref{dual1}, \eqref{dual2}, and \eqref{dual3}, we have, $\alg\geq \frac{\alpha}{\beta}\,\opt$. 
\end{lemma}
\begin{proof}{Proof.}
	Given $i\in I$ and $\umi\in[0,1]^{n-1}$, let $\opt_i$ denote the expected reward of $\opt$ from successfully matching $i$. Using Corollary \ref{basiceq}(ii) and Lemma \ref{expapprox}$(ii)$, we have, 
	\[\opt_i= (1\pm O(\sqrt{p_i}))\, r_i E_{\umi}\left[1-e^{-p(m^{\umi}_{\opt})}\right].\]
	Thus, given a feasible solution to the system, we have, 
	\[	(1-O(\sqrt{p_i}))\alpha\, \opt_i \leq	E_{\umi}\left[\theta^{\umi}_i + \sum_{t\mid (i,t)\in E} \lambda^{\umi}_{it}\right]\quad \forall i\in I\]
	Summing up these inequalities over $i\in I$, we get, 
		\begin{eqnarray*}
		(1-O(\sqrt{p}))\alpha\, \opt &&\leq\sum_{i\in I} E_{\umi}\left[\theta^{\umi}_i + \sum_{t\mid (i,t)\in E} \lambda^{\umi}_{it}\right],\\
		&&\leq \sum_{i\in I}E_{\umi}\left[\theta^{\umi}_i\right] + \sum_{(i,t)\in E} E_{\umi}\left[\lambda^{\umi}_{it}\right],\\
		&&\leq \beta \alg.
	\end{eqnarray*}
\hfill\Halmos 
\end{proof}

\subsection{Finding a Feasible Solution}\label{sec:fsb}
In this section, we perform a dual fitting i.e., using the matching generated by $\alg$ and $\opt$, we construct a feasible solution to the system given by \eqref{dual1}, \eqref{dual2}, and \eqref{dual3}.
 To define the path-based candidate solution, we couple the randomness in \alg\ and \opt\ by choosing the same thresholds $\ub$ for both algorithms. 
We start by introducing new terminology and notation that we use to define the candidate solution.
 \smallskip
 
 \noindent \textbf{Truncating $\opt$:} 
Consider a non-anticipative deterministic algorithm $\topt$ that executes \opt\ but leaves some arrivals unmatched, resulting in an effort threshold of at most $\frac{1}{\sqrt{p}}\,\, \forall i\in I$. For $i\in I$ and  $t\in T$, let $l_i(\opt,t)$ denote the sum of probabilities of all edges matched to $i$ (in \opt) prior to the match to $t$. Then, for every match $(i,t)$ made by \opt, the truncated algorithm $\topt$ matches $t$ as follows,
 \begin{enumerate}[1.]
 	\item If $l_i(\opt,t)\leq 1/\sqrt{p}$, then match $t$ to $i$.
 	\item If $l_i(\opt,t)> 1/\sqrt{p}$, then leave $t$ unmatched.
 \end{enumerate}

Observe that for every resource $i$ and (partial) sample path $\umi$, the effort threshold of $i$ in $\topt$ i.e., $p(m^{\umi}_{\topt})$, is at most $\frac{1}{\sqrt{p}}$. The next lemma states that truncation has a negligible effect on the expected total revenue. See Appendix \ref{appx:fsb} for a proof.

\begin{lemma}\label{trunc}
For every $i\in I$,	let $\opt_i$ and $\topt_i$ denote the expected total revenue from matching resource $i$ in \opt\ and in $\topt$ respectively. Then, we have, $\topt_i \geq (1- O(\sqrt{p})) \opt_i\,\,\, \forall i\in I$. 
\end{lemma}
%
%
%
\smallskip

 Let $\onee(A)$ be an indicator for event $A$.
Define,
\[tr(x)=\min\left\{x, \frac{1}{\sqrt{p}}\right\}\qquad \forall x\in\mathbb{R}.\]
Let $\topt_i(\umi,b)$ denote the set of arrivals matched to $i$ in $\topt$ on sample path given by thresholds $\umi$ and budget $b$ for resource $i$.
Recall that we use a function $g(\cdot)$ in Algorithm \ref{fullada} to compute the reduced rewards for resources. We define the candidate solution through functions $z^{\umi}_{it}:\mathbb{R}^+\times \mathbb{R}^+\to [0,1]$ and $v^{\umi}_i:\mathbb{R}^+\to [0,1]$. For every $i\in I, (i,t)\in E, \umi\in[0,1]^{n-1}$, let
\begin{eqnarray}
	&& z^{\umi}_{it}(x, y)=p_{it}r_i\, g\left(tr(p(m^{\umi}_\alg))\right)\times \onee\left(t\in \topt_i\left(\umi, x\right)\right)\times \onee\left( y>tr(p(m^{\umi}_\alg))\right) 	\label{z}\\
	&& v^{\umi}_i(x)=   r_i \int_{0}^{\min\{x,\, tr(p(m^{\umi}_\alg))\}} (1-g(y)) dy \label{v}.
\end{eqnarray}
{\color{black} At a high level, these functions split up the total reward of \alg. The function $z^{\umi}_{it}(x, y)$ is a lower bound on the reduced reward of resource $i$ seen by arrivals matched to $i$ in \opt. When the budget of $i$ in \alg\ is ``large", $i$ is not successfully matched and is available at every arrival. In such a scenario, $z^{\umi}_{it}(x, y)$ is non-zero at all arrivals $t\in \topt_i\left(\umi, x\right)$. 
} Now, our candidate solution is defined as,
\begin{eqnarray}
	\lambda^{\umi}_{it}= E_{b\sim Exp(1)}\left[z^{\umi}_{it}\left(b,b\right)\right]\quad \text{ and }\quad \theta^{\umi}_i=E_{b\sim Exp(1)}\left[v^{\umi}_i\left(b\right)\right].\label{lambda2}\label{theta2}
\end{eqnarray}
The solution is non-negative by definition, satisfying \eqref{dual3}. To show the validity of this candidate solution, we first show that inequality \eqref{dual2} is satisfied for a sufficiently small value of $\beta$ (Lemma \ref{dual2proof}). Then, we show that inequalities \eqref{dual1} are satisfied for a large enough value of $\alpha$ (Lemma \ref{generic}). We establish these claims for the following family of functions $g(\cdot)$. 
\smallskip

\noindent \textbf{Valid functions:}  Function $g:[0,\infty)\to [0,1]$ is valid if it satisfies the following properties,
\begin{enumerate}[(a)]
	\item $g(\cdot)$ is a strictly decreasing and differentiable function.
	\item There exists a constant $\eta\geq0$, such that, $g(x)\leq \big(1+\eta \sqrt{\epsilon}\big)\,g(x+\epsilon)\quad \forall \epsilon\leq 1, x\geq 0$.
	\item There exists a constant $c\geq 0$, such that, $ g(x)\leq \frac{c}{x} \quad \forall x>0$.
	\item $\min_{x\geq 0}f(x)=g(0)$, where,
	\[f(x)=1-e^{-x}[1-g(x)(x+1)]-\int_0^{x} g(z)e^{-z} dz.\]
	
	
\end{enumerate} 

\begin{lemma}\label{dual2proof}
	For any valid function $g$ and $p\leq 0.38$,
	the candidate solution given by \eqref{lambda2} satisfies constraint \eqref{dual2} with $\beta=1+O(\sqrt{p})$.
\end{lemma}
\begin{proof}{Proof.}
We first present an alternative way to write the expected total reward of $\alg$. 
Let $\alg(\ub)$ denote the matching generated by $\alg$ given thresholds $\ub$. Let $l^{\ub}_i(t)$ denote the sum of probabilities of edges in $\alg(\ub)$ that are incident on $i$ prior to $t$. For every edge $(i,t)\in \alg(\ub)$, let
	\begin{eqnarray*}
		&&f^{\ub}_{t}= r_i \int_{l^{\ub}_i(t)}^{l^{\ub}_i(t)+p_{it}} g(y) dy,\\
		&&h^{\ub}_{it}= r_i \int_{l^{\ub}_i(t)}^{l^{\ub}_i(t)+p_{it}} (1-g(y)) dy.
	\end{eqnarray*}
Then, for every $(i,t)\in\alg(\ub)$ we have, $f^{\ub}_{t}+h^{\ub}_{it}=p_{it} r_i$. We claim that,
\[\alg = E_{\ub}\left[\sum_{t\in T}f^{\ub}_{t}+\sum_{(i,t)\in \alg(\ub)}h^{\ub}_{it}\right].\]
The proof of this claim is a generalization of Lemma 2 in \cite{deb}. Consider a variant of the stochastic rewards setting where resource usage is stochastic but rewards are deterministic. Given a match $(i,t)$ in the variant, resource $i$ is consumed (unavailable for future matches) with probability $p_{it}$ but we earn a reward $p_{it}r_i$ with probability 1. 
By coupling the thresholds $\ub$ in this variant and the original setting, and using the linearity of expectation, we have that the overall reward of a non-anticipative algorithm $\mathcal{A}$ (such as $\alg$), is the same in both settings. Observe that, $E_{\ub}\left[\sum_{t\in T}f^{\ub}_{t}+\sum_{(i,t)\in \alg(\ub)}h^{\ub}_{it}\right]$ is the expected reward of $\alg$ on the variant problem, giving us the desired. 
Now, to prove the main claim it suffices to show that,
\begin{eqnarray}
&&\sum_{i\in I} E_{\umi}[\theta^{\umi}_{i}]\leq (1+O(\sqrt{p}))\,  \sum_{(i,t)\in \alg(\ub)} E_{\ub}[h^{\ub}_{it}] 
,\label{ub1}\\
&&\sum_{(i,t)\in E} E_{\umi}\left[\lambda^{\umi}_{it}\right]
\leq (1+O(\sqrt{p}))\,E_{\ub}\left[\sum_{t\in T} f^{\ub}_{t}\right]+O(\sqrt{p})\,\alg. \label{ub2}
\end{eqnarray}
\smallskip

\noindent \emph{Proof of \eqref{ub1}:} 
 For every $i\in I$ and $\umi\in[0,1]^{n-1}$, observe that $v^{\umi}_i(x)$ meets the criteria in Lemma \ref{singfunc} -- it is integrable and non-decreasing in $x$, has a constant value for $x\geq \frac{1}{\sqrt{p_i}}$, and $v^{\umi}_i(0)=0$.  
Applying the lemma we get,
\[\theta^{\umi}_i= (1\pm O(\sqrt{p_i}))\, E_{u_i}\left[v^{\umi}_i\left(G^{\umi}_{\alg}(u_i)\right)\right]\qquad i\in I,\, \umi\in[0,1]^{n-1}.\] 
Now, to prove \eqref{ub1}, it suffices to show that $\sum_{i\in I} v^{\umi}_{i}\left(G^{\umi}_{\alg}(u_i)\right) \leq \sum_{i\in I} \sum_{t \in \alg_i(\ub)} h^{\ub}_{it}\quad \forall \ub\in[0,1]^{n}$. In fact, we show the following stronger statement. 
Consider an arbitrary $i\in I$ and $\ub\in[0,1]^{n}$. Let $\alg_i(\ub)$ denote the set of arrivals matched to $i$ in matching $\alg(\ub)$.  Then, 
\begin{eqnarray*} 
	\sum_{t \in \alg_i(\ub)} h^{\ub}_{it} &= &r_i \sum_{t \in \alg_i(\ub)} \int_{l^{\ub}_i(t)}^{l^{\ub}_i(t)+p_{it}} (1-g(y)) dy,\\
	&= & r_i \int_{0}^{\min\left\{G^{\umi}_{\alg}(u_i),\,  p(m^{\umi}_{\alg})\right\}}
	(1-g(y)) dy,\\
	&\geq & v^{\umi}_{i} \left(G^{\umi}_{\alg}(u_i)\right).
	\end{eqnarray*}
 The first equality and the inequality follow by definition. Second equality follows from the fact that $\underset{t\in \alg_i(\ub)}{\max} l^{\ub}_i(t)+p_{it}$, which is the sum of probabilities of all edges incident on $i$ in $\alg(\ub)$, is equal to $\min\left\{G^{\umi}_{\alg}(u_i), p(m^{\umi}_{\alg})\right\}$. This completes the proof of \eqref{ub1}. 
\smallskip

\noindent \emph{Proof of \eqref{ub2}:} The proof consists of two parts. First, we show that,
\[\lambda^{\umi}_{it} \leq (1+O(\sqrt{p_i})) E_{u_i}\left[z^{\umi}_{it}\left(G^{\umi}_{\opt}(u_i), G^{\umi}_\alg(u_i)\right)\right]\quad \forall (i,t)\in E, \umi\in[0,1]^{n-1}.\]
Then, we show,
\[\sum_{(i,t)\in E} E_{\ub}\left[z^{\umi}_{it}\left(G^{\umi}_{\opt}(u_i),G^{\umi}_{\alg}(u_i)\right)\right]
\leq (1+O(\sqrt{p}))\,E_{\ub}\left[\sum_{t\in T} f^{\ub}_{t}\right]+O(\sqrt{p})\,\alg. \]
Combining these inequalities gives us the desired.

Given an edge $(i,t)$ and (partial) sample path $\umi$, observe that,
\[z^{\umi}_{it}\left(G^{\umi}_{\opt}(u_i), G^{\umi}_\alg(u_i)\right)=c\, \onee(A(u_i))\, \onee(B(u_i)),\]
 here $c$ is a constant independent of $u_i$, $A(u_i)$ corresponds to the event that $t\in \topt_{i}(\umi, G^{\umi}_{\opt}(u_i))$, and $B(u_i)$ is the event, $G^{\umi}_{\alg}(u_i)>tr(p(m^{\umi}_\alg))$.  
 From Lemma \ref{basicu}$(ii)$ and Corollary \ref{basiceq}$(i)$, we have that for some $u_A, u_B\in[0,1]$ (independent of $u_i$), event $A(u_i)$ and event $B(u_i)$ are equivalent to events $u_i \geq u_A$ and $u_i \geq u_B$ respectively. For brevity, let \[Z^{\umi}_{it}(u_i)=z^{\umi}_{it}\left(G^{\umi}_{\opt}(u_i),G^{\umi}_{\alg}(u_i)\right)=c\, \onee(u_i>u_A)\, \onee(u_i>u_B).\] 
 Define random variables $X_A(u_i)= c\, \onee(u_i>u_A)$ and $X_B(u_i)= c\, \onee(u_i>u_B)$. 
Then,
\begin{eqnarray*}
	E_{u_i}[z^{\umi}_{it}(u_i)]&= &\onee(u_A\geq u_B)E_{u_i}[X_A(u_i)]+\onee(u_A<u_B)E_{u_i}[X_B(u_i)]\\
&=	&\left(1\pm O(\sqrt{p_i})\right)\, c\,\Big[\onee(u_A\geq u_B)E_{b\sim Exp(1)}[\onee(b>G^{\umi}_{\opt}(u_A))]\\
&&\qquad \qquad\quad\qquad+\onee(u_A<u_B)E_{b\sim Exp(1)}[\onee(b>G^{\umi}_{\alg}(u_B))]\Big]\\
&\geq &\left(1- O(\sqrt{p_i})\right)\, E_{b\sim Exp(1)}[z^{\umi}_{it}(b,b)]\\
&= &\left(1- O(\sqrt{p_i})\right)\, \lambda^{\umi}_{it}.
\end{eqnarray*}
The second equality follows from Lemma \ref{funcapx}. The inequality follows from the fact that for every $b\geq 0$, both $\onee(b>G^{\umi}_{\opt}(u_A))$ and $\onee(b>G^{\umi}_{\alg}(u_B))$ are lower bounded by $\onee(b>G^{\umi}_{\opt}(u_A))\onee(b>G^{\umi}_{\alg}(u_B))$.




Now, we prove the desired upper bound on $\sum_{(i,t)\in E} E_{\ub}\left[Z^{\umi}_{it}(u_i)\right]$. Given sample path $\ub$, we partition resources $i\in I$ into two categories - \emph{normal} and \emph{extreme}. A resource $i$ is \emph{normal} if the effort threshold $p(m^{\umi}_{\alg})\leq 1/\sqrt{p}$. Let $I_1(\ub)$ denote the set of normal resources. Let $I_2(\ub)=I\backslash I_1(\ub)$, denote the set of extreme resources. 

Consider an arbitrary sample path $\ub$ and an arrival $t\in T$. Let $i^*$ denote the normal resource (if any) matched to $t$ in the matching $\topt(\ub)$. Then,
\[\sum_{i\mid\, i\in I_1(\ub),\, (i,t)\in E} Z^{\umi}_{it} (u_i)= r_{i^*}\,\, p_{i^* t}\,\, g(p(m^{u^{\text{-}i^*}}_\alg))\,\, \onee\left(G^{u^{\text{-}i^*}}_\alg (u_i) > p(m^{u^{\text{-}i^*}}_\alg)\right)\,\leq\, (1+O(\sqrt{p}))\,f^{\ub}_t, \]
here the last inequality follows from the fact that \alg\ is greedy w.r.t. expected reduced rewards and property (b) of valid functions. Thus, 
	\[\sum_{(i,t)\in E, i\in I_1(\ub)}Z^{\umi}_{it}(u_i)\,=\,\sum_{t\in T}\,\, \sum_{i\mid\, i\in I_1(\ub),\, (i,t)\in E} Z^{\umi}_{it} (u_i)\,\leq\, (1+O(\sqrt{p}))\, \sum_{t\in T}f^{\ub}_t\quad \forall \ub\in[0,1]^n.\]
For extreme resources, we have the following upper bound for every $\umi\in[0,1]^{n-1}$,  
\begin{eqnarray*}
	\sum_{(i,t) \in E,\, i\in I_2(\ub)}Z^{\umi}_{it}(u_i)&= & \sum_{i\in I_2(\ub)}  \left[r_i g\left(\frac{1}{\sqrt{p}}\right)\times \onee\left(G^{\umi}_{\alg}(u_i)>\frac{1}{\sqrt{p}}\right)\times \sum_{t\in \topt_i(\umi,G^{\umi}_{\opt}(u_i))} p_{it}\right]\\
	&\overset{(A)}{\leq} &\sum_{i\in I_2(\ub)}  r_i  \frac{g(\frac{1}{\sqrt{p}})}{\sqrt{p}}\,\, \onee\left(G^{\umi}_{\alg}(u_i)>\frac{1}{\sqrt{p}}\right)\\
	& \overset{(B)}{\leq} &O(1)\sum_{i\in I_2(\ub)}  r_i \, \onee\left(G^{\umi}_{\alg}(u_i)>\frac{1}{\sqrt{p}}\right).
\end{eqnarray*} 
Inequality $(A)$ follows from the fact that $\sum_{t\in \topt_i(\umi, G^{\umi}_{\opt}(u_i))} p_{it}\leq tr(G^{\umi}_{\opt}(u_i))\leq \frac{1}{\sqrt{p}}$. Inequality $(B)$ follows from  property $(c)$ of valid functions. Combining the upper bounds for normal and extreme resources, we get,
\begin{eqnarray*}
	E_{\ub}\left[\sum_{(i,t) \in E}Z^{\umi}_{it}(u_i)\right]&=& E_{\ub}\left[\sum_{(i,t) \in E,\, i\in I_1(\ub)}Z^{\umi}_{it}(u_i)\right] +E_{\ub}\left[\sum_{(i,t) \in E,\, i\in I_2(\ub)}Z^{\umi}_{it}(u_i)\right]\\
	&\leq & \left(1+O(\sqrt{p})\right)\, E_{\ub}\left[\sum_t f^{\ub}_t\right]\\
	&& + O(1)\sum_{i\in I}r_iE_{\ub}\left[\onee(i\in I_2(\ub)) \times \onee\left(G^{\umi}_{\alg}(u_i)>\frac{1}{\sqrt{p}}\right)\right]
\end{eqnarray*}
It remains to show that $\sum_{i\in I}r_iE_{\ub}\left[\onee(i\in I_2(\ub)) \times \onee\left(G^{\umi}_{\alg}(u_i)>\frac{1}{\sqrt{p}}\right)\right]\,\leq\, O(\sqrt{p})\, \alg $. 
\begin{eqnarray*}
	&&\sum_{i\in I}r_i\,E_{\ub}\left[\onee(i\in I_2(\ub)) \times \onee\left(G^{\umi}_{\alg}(u_i)>\frac{1}{\sqrt{p}}\right)\right]\\
	&&= \sum_{i\in I} r_i\, E_{\umi}\left[\onee\left(p(m^{\umi}_{\alg})>\frac{1}{\sqrt{p}}\right)\times E_{u_i}\left[\onee\left(G^{\umi}_{\alg}(u_i)>\frac{1}{\sqrt{p}}\right)\right]\right]\\
	&&\leq \sum_{i\in I} r_i\,  E_{\umi}\left[\onee\left(p(m^{\umi}_{\alg})>\frac{1}{\sqrt{p}}\right)\right]\times O(\sqrt{p})\\
	&&\leq O(\sqrt{p}) \sum_{i\in I} r_i\, E_{\umi}\left[\onee\left(p(m^{\umi}_{\alg})>\frac{1}{\sqrt{p}}\right)\times \frac{E_{u_i}\left[\onee\left(G^{\umi}_{\alg}(u_i)\leq\frac{1}{\sqrt{p}}\right)\right]}{1-O(\sqrt{p})}\right]\\
	&&\leq O(\sqrt{p})\sum_{i\in I} r_i\, E_{\umi}[E_{u_i}[\onee(G^{\umi}_{\alg}(u_i)<p(m^{\umi}_{\alg}))]]\\
	&&= O(\sqrt{p})\, \alg,
\end{eqnarray*}
here we used the fact that $F^{\umi}_{\alg}(x)=(1\pm O(\sqrt{p})) (1-e^{-x})$ from Lemma \ref{expapprox}$(ii)$ , and the approximation, $1-e^{-\frac{1}{\sqrt{p}}}\geq(1-O(\sqrt{p}))$.
\hfill\Halmos
 \end{proof}

\begin{lemma}\label{generic} 
	Given a valid function $g$, 
	inequalities \eqref{dual1} are satisfied with $\alpha= (1-O(\sqrt{p}))\, g(0)$. 
\end{lemma}
\begin{proof}{Proof.}
	Fix an arbitrary resource $i$ sample path $\umi$. 
To establish \eqref{dual1}, we show that,
\begin{equation}\label{dualtrunc}
	\theta^{\umi}_i+ \sum_{t\mid (i,t)\in E} \lambda^{\umi}_{it}\, \geq\, g(0)\, \left(1-e^{-p\left(m^{\umi}_{\topt}\right)}\right). \end{equation}
Then, using Lemma \ref{trunc} completes the proof.

To prove \eqref{dualtrunc}, we first find lower bounds on $\theta^{\umi}_i$ and $\sum_{t\mid (i,t)\in E} \lambda^{\umi}_{it}$ separately. For brevity, we omit $\umi$ from notation.  In particular, let $\theta_i$ and $\lambda_{it}$ denote the candidate solution $\theta^{\umi}_i,\lambda^{\umi}_{it}$. Let $\topt_i(b)$ denote the set $\topt_i(\umi, b)$ of arrivals matched to $i$ in $\topt$ on sample path given by $(\umi, b)$. Let $\tau_O$ denote the effort threshold $p(m^{\umi}_{\topt})=tr(p(m^{\umi}_{\topt}))$. Similarly, let $\tau_A$ denote the effort threshold $tr(p(m^{\umi}_{\alg}))$. 
	Let $G(z)=\int^z_0 g(x) dx$. Then, we have, 
	\begin{eqnarray}
	\theta_i&&= r_i\,E_{b\sim Exp(1)}\left[ \int_{0}^{\min\{b,\,\tau_A\}} (1-g(y)) dy\right]\nonumber\\
 &&= r_i\Big[\int_0^{\tau_A} \Big(z-\int_{0}^{z} g(y)dy\Big)e^{-z} dz + \int_{\tau_A}^{+\infty}\Big(\tau_A-\int_{0}^{\tau_A} g(y)dy \Big) e^{-z} dz\Big]\nonumber \\
	&&= r_i\Big[-\tau_A e^{-\tau_A} +1-e^{-\tau_A} -\int_0^{\tau_A} (G(z)-G(0)) e^{-z} dz \nonumber\\
	&&\qquad  +\tau_A e^{-\tau_A} -(G(\tau_A)-G(0))e^{-\tau_A}  \Big]\nonumber \\
	&&= r_i \Big[1-e^{-\tau_A}-\int_0^{\tau_A} g(z)e^{-z} dz\Big]. 
	\label{thetahatdef2}
	\end{eqnarray}
	\begin{eqnarray}
	\sum_{t\mid (i,t)\in E}\lambda_{it}&&= r_i\,g(\tau_A)\,\sum_{t\mid (i,t)\in E}p_{it}\,E_{b\sim Exp(1)}\left[ \onee\left(t\in \topt_i(b)\right)\,\, \onee\left( b>\tau_A\right) \right]\nonumber\\
		&&=r_i\,g(\tau_A)\, \int_{\tau_A}^{\infty}\Big(\sum_{ t\in \topt_i( b)}p_{it} \Big) e^{-z}dz\nonumber\\
	&&\geq r_i\,g(\tau_A)\, \int_{\tau_A}^{\infty} \min\{\tau_O,z\} e^{-z} dz\nonumber\\
	&&=r_i\, g(\tau_A)\, \Big[\int_{\tau_A}^{\max\{\tau_O,\tau_A\}} z e^{-z} dz + \tau \int_{\max\{\tau_O,\tau_A\}}^{\infty} e^{-z} dz \Big]\nonumber\\
	&& = r_i\, g(\tau_A)\times \begin{cases}\label{lambdalb}
	\tau_O e^{-\tau_A} &\quad \tau_O \leq \tau_A,\\
	\tau_Ae^{-\tau_A}+e^{-\tau_A}-e^{-\tau_O} &\quad \tau_O> \tau_A.
	\end{cases}
	\end{eqnarray}
Combining \eqref{thetahatdef2} and \eqref{lambdalb}, consider the function,
	\begin{eqnarray*}
		H(\tau_O,\tau_A)&&:= 1-e^{-\tau_A}-\int_0^{\tau_A} g(z)e^{-z} dz + g(\tau_A)\times \begin{cases}
			\tau_O e^{-\tau_A} &\quad \tau_O \leq \tau_A,\\
			\tau_Ae^{-\tau_A}+e^{-\tau_A}-e^{-\tau_O} &\quad \tau_O> \tau_A,
		\end{cases}
	\end{eqnarray*}
	and its derivative,
	\begin{eqnarray*} 
		\frac{\partial H(\tau_O,\tau_A)}{\partial \tau_O}&&\leq g(\tau_A)\times \begin{cases}
			e^{-\tau_A} &\quad \tau_O \leq \tau_A\\
			e^{-\tau_O} &\quad \tau_O> \tau_A \end{cases}\\
		&& 
		<\, g(0)e^{-\tau_O}\quad \forall \tau_A>0,
	\end{eqnarray*}
	here the last inequality uses the fact that $g$ is strictly decreasing (property (a) of valid functions). To prove \eqref{dualtrunc}, it suffices to lower bound the function,
	\[h(x,y)=\frac{H(x,y)}{(1-e^{-x})}\quad \forall  x,y\geq 0.\]
	{\color{black} Using property (d) of valid functions, we have that $\min_{y\geq 0} h(\infty,y)= \min_{x\geq 0} f(x)= g(0)$. Also, $h(x,0)=g(0)$ for all $x\geq 0$. Thus, if there exists a pair of values $x,y$ such that  $h(x,y)<g(0)$, then $x$ must be finite and $y$ must be strictly larger than 0. Consider one such pair $x_1,y_1$. Since $	\frac{\partial H(x,y_1)}{\partial x}<  g(0)e^{-x}=g(0)\, \frac{d(1-e^{-x})}{dx}$, 
	 we have that,
	\[ h(\infty,y_1)\,< \, h(x_1,y_1)\,<\, g(0).\]
	This contradicts the fact that $\min_{y\geq 0} h(\infty,y)= g(0)$. Thus, $h(x,y)\geq g(0)$ for all values $x,y\geq 0$.} 

\hfill\Halmos 
\end{proof}

\begin{proof}{Proof of Theorem \ref{mainh}.}
	Given a valid function $g$ and $p\leq 0.38$, from Lemma \ref{dual2proof} and Lemma \ref{generic} we have that the candidate solution given by \eqref{lambda2} is a feasible solution to the LP free system with $\alpha= (1-O(\sqrt{p}))\, g(0)$ and $\beta=1+O(\sqrt{p})$. Then, using Lemma \ref{final}, we have that $\alg$ is $(1-O(\sqrt{p}))\, g(0)$ competitive.
	
	It remains to find a valid function such that $g(0)\geq 0.596$. 
 Consider the function $g(x)=e^{x+1}\int_{x+1}^{\infty} \frac{e^{-y}}{y} dy$. This is a strictly decreasing and differentiable function (property (a) of valid functions) with $0.5963\leq g(0)\leq0.5964$ (bounds numerically evaluated).
It can be verified that	properties (b) and (c) follow from, 
	\[\frac{1-e^{-x-1}}{2x+2} \,=\, \frac{e^{x+1}}{2x+2}\int_{x+1}^{2x+2} e^{-y} dy\,\leq\,\, g(x)\,\,\leq\, \frac{e^{x+1}}{x+1}\int_{x+1}^{\infty} e^{-y} dy\,=\, \frac{1}{x+1}\quad \forall x\geq 0.\]
	
	 Finally, we also need to show property (d) i.e., $\min_{x\geq 0}f(x)=g(0)$. 
	By definition of $g(\cdot)$ we have, $g(x)-g'(x)=\frac{1}{x+1}$. Thus, $\frac{df(x)}{dx}=0$, which gives us the desired. 
\Halmos \end{proof}
\begin{corollary}
	For $g(x)=\beta_1 \frac{1}{\beta_2x+1}$, with $\beta_1=0.588$, $\beta_2=0.575$, and $p\leq 0.38$, Algorithm \ref{fullada} is $(1-O(\sqrt{p}))\, 0.588$ competitive.
\end{corollary}
\begin{proof}{Proof.}
	Notice that $g(\cdot)$ is strictly decreasing for $x\geq 0$ and $g(0)=\beta_1=0.588$. Property (b) and (c) follow by definition. One can verify analytically/numerically that $0$ is a minimizer of $f(\cdot)$ (property (d)). 
\hfill\Halmos 
\end{proof}

Similarly, it can be verified that the following guarantee holds for the exponential function $g(x)=e^{-x}$ originally proposed in \cite{deb}.
\begin{corollary}
	For $g(x)=\beta_1 e^{-\beta_2 x}$, with $\beta_1=0.581$, $\beta_2=0.535$, and $p\leq 0.38$, Algorithm \ref{fullada} is $(1-O(\sqrt{p}))\, 0.581$ competitive. 
\end{corollary}

\noindent \emph{Remarks:} In case of vanishing probabilities, the classical upper bound of $(1-1/e)$ does not necessarily hold when comparing against fully offline or clairvoyant. Unlike the decomposable case, which includes the classical case of unit probabilities, the vanishing probabilities case does not immediately include any previous setting with the $(1-1/e)$ barrier. More concretely, recall that from the thresholding process perspective, the marginal increase in objective value obtained by matching $i$ to $t$ is approximately $e^{-l_i(t)}p_{it}$. This value decreases over time if $i$ is unsuccessfully matched many times and $l_i(t)$ increases. As a consequence of these diminishing marginal gains, the hard instances that lead to a $(1-1/e)$ upper bound in deterministic settings are not as detrimental here. 


\section{Equivalence of PBP and Expectation LP for Large Inventory}\label{asymptote}
So far, we focused on the case where there is exactly one unit of every offline vertex $i\in I$. An instance where we have (possibly) multiple units $c_i\geq 1\,\, \forall i\in I$, is a special case of the single unit setting. Indeed, recall that under the assumption that $c_i\to \infty$ for every $i$, the algorithm of \cite{msvv} is $(1-1/e)$ competitive against the expectation based LP for arbitrary probabilities. This raises a natural question: \emph{Is the path based program PBP equivalent to the expectation based LP for large inventory?} In this section, we formally show this to be the case. In fact, we show more strongly that for large inventory, both $OPT(LP)$ and $OPT(PBP)$ are approximately equal to the expected reward of offline (\opt).


\begin{theorem}Let $c_{min}=\min_{i\in I}c_i$. 
	For $c_{min}\geq 2$, we have,
	\[OPT(LP)\Big(1-O\Big(\sqrt{\frac{\log c_{min}}{c_{min}}}\Big)\Big)\leq OPT \leq OPT(PBP)\leq OPT(LP).\] 
	Hence, for $c_{min}\to \infty$, we have, $\opt\to OPT(LP)$.
\end{theorem}
\begin{proof}{Proof.}
	Let $\{x_{it}\}_{(i,t)\in E}$ be a feasible solution for the expectation based LP. Consider the non-anticipative offline algorithm that for some carefully chosen value $\delta>0$ (to be decided later), matches every arrival $t\in T$ (in order) by independently sampling a resource according to the distribution $\{\frac{x_{it}}{1+\delta}\}_{i\in I}$.  If the sampled resource is unvailable, the arrival is not matched. We call this the independent rounding (IR) algorithm. Let $IR(\delta)$ denote the expected total reward of this algorithm. Note that $IR (\delta)\leq \opt\,\, \forall \delta\geq 0$. We show that for $\delta=\sqrt{\frac{2\log c_{min}}{c_{min}}}$,
	\begin{equation}\label{ir}
	OPT(LP)\Big(1-O\Big(\sqrt{\frac{\log c_{min}}{c_{min}}}\Big)\Big)\leq IR \left(\sqrt{\frac{2\log c_{min}}{c_{min}}}\right),
	\end{equation}
	 which gives us the desired. To prove \eqref{ir}, it suffices to show that in IR, for every edge $(i,t)\in E$, resource $i$ is successfully matched to $t$ w.p.\ at least $p_{it}\,x_{it}\,\left(1-O\Big(\sqrt{\frac{\log c_{min}}{c_{min}}}\Big)\right)$. To this end, consider the following indicator random variables for every $(i,t)\in E$, (i) $A(i,t)$ to indicate if $i$ is available at $t$, (ii) $\onee(i,t)$ to indicate the success/failure (stochastic reward) of the edge, and (iii) $X_{it}$ to denote if $i$ is sampled at $t$. 

	Given an arbitrary edge $(i,t)$, the probability that $i$ is successfully matched to $t$ is 
	\[p_{it}\,\, \frac{x_{it}}{1+\sqrt{\frac{2\log c_{min}}{c_{min}}}}\, \mathbb{P}\left( \sum_{\tau\leq t } A(i,\tau)\,\, \onee(i,\tau)\, X_{i\tau} \leq c_i \right),\]
where,
	\begin{eqnarray*}
		\mathbb{P}\left( \sum_{\tau\leq t } A(i,\tau)\, \onee(i,\tau)\, X_{i\tau} \leq c_i \right) &\geq &\mathbb{P}\left( \sum_{\tau\leq t } \onee(i,\tau)\, X_{i\tau} \leq c_i \right)\\
		&= & 1-\mathbb{P}\left( \sum_{\tau\leq t } \onee(i,\tau)\, X_{i\tau} > c_i \right)\\
		&\geq &1-\frac{1}{\sqrt{c_{min}}}.
	\end{eqnarray*}
The last inequality follows from the standard Chernoff bound for independent Bernoulli random variables. To see this, consider random variables $\{Y_{\tau}\}_{\tau=1}^{t}:=\{\onee(i,\tau)X_{i\tau}\}_{\tau=1}^t$. These are independent Bernoulli random variables (by definition). From the feasibility of $\{x_{it}\}_{(i,t)\in E}$, the total mean, $E[\sum_{\tau\leq t} Y_{\tau}]\leq \frac{c_i}{1+\sqrt{\frac{\log c_{min}}{c_{min}}}}$. Thus,
\[\mathbb{P}\left(\sum_{\tau\leq t}Y_{\tau}\geq c_i\right) \leq e^{-\frac{\frac{2\log c_{min}}{c_{min}}}{3}c_i}\leq 
\frac{1}{\sqrt{c_i}}. \]

\hfill\Halmos 
\end{proof}

\section{Conclusion, Open Problems, and Further Work}\label{open}
In this paper, we consider a vertex-weighted version of the problem of online matching with stochastic rewards with the goal of designing online algorithms to compete against offline solutions that know the entire sequence of arrivals in advance but learn the stochastic outcomes after matching decisions are made. 
We show that when probabilities $p_{it}$ decompose as $p_{i}\times p_t$ for every edge $(i,t)$, a natural generalization of the Perturbed-Greedy algorithm gives the best possible competitive guarantee of $(1-1/e)$. When probabilities are small but otherwise fully heterogeneous, we show that a deterministic algorithm that generalizes the Fully-Adaptive algorithm suggested in \cite{deb}, is $0.596$ competitive. 

Our analysis for decomposable probabilities involves a novel path based program that gives a tighter bound on the optimal value of offline. 
For the second case, we develop a new LP free approach where it suffices to find a solution to a certain linear system that compares an online algorithm directly to the actions of offline. We propose a novel candidate solution for this linear system and by generalizing an alternative perspective of the stochastic reward process introduced in \cite{deb}, we show that the proposed candidate indeed satisfies the linear system. To the best of our knowledge, our results are the first to give a means to obtain provably tighter bound on clairvoyant and fully offline beyond the standard expectation LP. Finally, we show that for arbitrary probabilities but large inventory (many copies of each offline vertex), the path based program and the expectation LP are asymptotically equivalent to clairvoyant.

Beating $1/2$ without any assumptions on probabilities still remains open for this problem. Similarly, in the more general setting of online assortments, achieving a competitive guarantee better than $1/2$ remains open (except for large inventory, see \cite{negin}). A more immediate line of investigation would be to determine the best one could do in the vanishing probabilities case. Here it may, in fact, be possible to achieve a guarantee better than $(1-1/e)$. We believe that the ideas of analysis developed here may be useful more broadly in other settings when comparing against offline benchmarks in the presence of stochastic elements. 
Indeed, \cite{reuseopt} further generalize the LP free framework, developed 
in Section \ref{sec:vanish2}, to circumvent issues with LP based primal-dual analysis for a very different type of post-allocation stochasticity\footnote{They consider the case of stochastic reusability, where resources are used for a stochastic duration and then returned for reallocation. The duration is realized only when the resource is returned for reallocation, making it a form of post allocation stochasticity.}.

%
%
%
 \begin{APPENDICES}
\section{Omitted Proofs}\label{appx:omitted}
\subsection{Missing Proofs from Section \ref{sec:vanish1}} \label{appx:vanish1}
\begin{repeatlemma}[Lemma \ref{same}.]
	Consider a non-anticipative and deterministic algorithm $\mcala$. Let $\mcala(\omega)$ denote the matching output by $\mcala$ on sample path $\omega\in \Omega$. 
	Similarly, let $\mcala(\ub)$ denote the matching output by $\mcala$ on sample path $\ub\in[0,1]^{n}$ in the threshold viewpoint. For every matching $\mathcal{M}$, we have $\mathbb{P}(\mcala(\omega)=\mathcal{M})=\mathbb{P}(\mcala(\ub)=\mathcal{M})$.
\end{repeatlemma}
\begin{proof}{Proof.}
	Given a matching $\mathcal{M}$ in the original viewpoint, consider an arbitrary resource $i\in I$ and let $\mathcal{M}_i$ denote the ordered set of arrivals matched to $i$. Arbitrarily fix the randomness in all edges not incident on $i$. Let $\omega_{-i}$ denote this (partial) sample path. Conditioned on $\omega_{-i}$, let $p_{\mathcal{M}_i}$ denote the probability (w.r.t.\ randomness in edges incident on $i$) that the $\mathcal{M}_i$ is the set of arrivals matched to $i$ by $\mathcal{A}$. 
	
	Now, conditioned on $\omega_{-i}$, consider the threshold viewpoint just for resource $i$. Generate a random threshold $u_i$ for $i$. By definition of the threshold viewpoint, the probability (w.r.t. randomness in $u_i$) that $\mathcal{A}$ matches $i$ to the ordered set of arrivals $\mathcal{M}_i$ in the threshold viewpoint, is exactly $p_{\mathcal{M}_i}$. Therefore, we have that sample paths $(\omega_{-i}, u_i)$, where the success/failure of matches made to $i$ are decided by the threshold viewpoint, are stochastically equivalent to sample paths $\omega$. Iteratively applying this argument to covert the sample path of all resources (one at a time), we have that sample paths $\ub$ and $\omega$ are stochastically equivalent.  
	\hfill\Halmos
\end{proof}
\begin{repeatlemma}[Lemma \ref{basicu}.]
	Consider $i\in I$, $\umi\in[0,1]^{n-1}$, and a deterministic non-anticipative algorithm $\mcala$. Let $\mcala^{\umi}_i(u_i)$ denote the ordered set of arrivals matched to $i$ in the matching $\mcala(\ui,\umi)$. 
	\begin{enumerate}[(i)]
		\item For every $k\in[m^{\umi}_{\mcala}]$, $i$ is matched to $t_k\in \mcala^{\umi}_i(1)$ if and only if $u_i>1-\prod_{j\leq k-1} (1-p_{it_j}) $. Thus, conditioned on $\umi$, $i$ is matched to $t_k$ with probability $\prod_{j\leq k-1} (1-p_{it_j})$.
		\item  For every $u_i\in[0,1]$, there exist some $k\in[m^{\umi}_{\mcala}]$ such that $\mcala^{\umi}_i(u_i)=\{t_1,\cdots,t_k\}$. Further,  $\mcala^{\umi}_i(u_i)\subseteq \mcala^{\umi}_i(u'_i)$ for every $0\leq u_i\leq u'_i\leq 1$. 
	\end{enumerate}
\end{repeatlemma}
\begin{proof}{Proof.}
	Given resource $i$ and thresholds $\umi$, let $\mathcal{A}^{\umi}_{i}(1)$ denote the ordered set of arrivals matched to $i$ when $u_i=1$. From the non-anticipative nature of $\mathcal{A}$, we have that for every threshold $u_i\in[0,1]$, $\mcala$ sequentially matches $i$ to the arrivals in $\mathcal{A}^{\umi}_{i}(1)$ (in order), until one of the following occurs (i) The probability of success exceeds the threshold $u_i$ and $i$ is successfully matched to an arrival in $\mathcal{A}^{\umi}_{i}(1)$, or, (ii) All arrivals in $\mathcal{A}^{\umi}_{i}(1)$ are unsuccessfully matched to $i$ and $i$ is not matched to any other arrivals.  
	
	Therefore, for every arrival $t_k\in \mathcal{A}^{\umi}_{i}(1)$, $i$ is matched to $t_k$ if and only if $u_i$ is large enough that the probability of one of $i$'s previous matches succeeding is smaller than $u_i$ i.e., $u_i>1-\prod_{j\leq k-1}(1-p_{it_j})$. This proves $(i)$. Further, this shows that if $i$ is matched to an arrival $t$ for some threshold value $u_i$, then $i$ is matched to $t$ for every threshold value $u'_i\geq u_i$. This proves part $(ii)$. 
	\hfill 	\Halmos 
\end{proof}
	\subsection{Missing Proofs from Section \ref{sec:trunc}} \label{appx:trunc}
	\begin{repeatlemma}[Lemma \ref{expapprox}.]
		Given a non-anticipative deterministic algorithm $\mathcal{A}$, resource $i\in I$, threshold vector $\umi\in[0,1]^{n-1}$ and $p\leq 0.38$, 
		we have, 
		\begin{enumerate}[(i)]
			\item $1-F^{\umi}_{\mcala}(x) = (1\pm O(\sqrt{p_i}))
			\,\, e^{-x}\quad \forall x\in[0,\frac{1}{\sqrt{p_i}}].$
			\item $ F^{\umi}_{\mcala}(x) = (1\pm 
			O(\sqrt{p_i}))\,\, (1-e^{-x})\quad \forall x\geq 0.$
		\end{enumerate} 
	\end{repeatlemma}
	\begin{proof}{Proof.}
		Consider an arbitrary value $x\geq 0$. 
		Let $\{p(1),\cdots, p(k)\}$ denote the set of all points $\leq x$ in the support of $F^{\umi}_{\mcala}$. Note that $x-p(k)$ can be at most $p_i$. Using the fact that $1-z\leq e^{-z}\leq 1$ for $z\geq 0$, we have that, $e^{-p(k)}=(1\pm O(p_i))\, e^{-x}$. Let $p(0)=0$ and $p_{it_j}=p(j)-p(j-1)\quad \forall j\in[k]$.
		\smallskip
		
		\noindent \emph{Proof of $(i)$:} By definition of distribution $F^{\umi}_{\mcala}$, we have, $1-F^{\umi}_{\mcala}(x)=\prod_{j=1}^k (1-p_{it_j})$. We show that for $x\in [0,\frac{1}{\sqrt{p_i}}]$, the product $\prod_{j=1}^k (1-p_{it_j})$ is well approximated by $e^{-p(k)}$. Then, the fact that $e^{-p(k)}=(1\pm O(p_i)) e^{-x}$, gives us the desired. 
		\begin{eqnarray*}
			1\leq \frac{e^{-\sum_{j=1}^k p_{it_j}}}{\prod_{j=1}^k (1-p_{it_j})}&&\overset{(a)}{\leq} \frac{\prod_{j=1}^k (1- p_{it_j}+p_{it_j}^2)}{\prod_{j=1}^k (1-p_{it_j})},\\
			&&\overset{}{=} \frac{\prod_{j=1}^k (1-p_{it_j})\prod_{j=1}^k\big(1+ \frac{p^2_{it_j}}{1-p_i}\big)}{\prod_{j=1}^k (1-p_{it_j})},\\
			&&\overset{(b)}{\leq} \prod_{j=1}^k e^{\frac{p_{i}}{1-p_i}p_{it_j}}\,\, \leq\quad e^{\frac{p_{i}}{1-p_i}\left(\frac{1}{\sqrt{p_i}}\right)},\\
			&&\overset{(c)}{\leq} 1+O\left(\sqrt{p_i}\right),
		\end{eqnarray*}
		here $(a)$ follows from the inequality, $e^{-z}\leq 1-z+z^2$ for $z\geq 0$. Inequality $(b)$ follows from the fact that $1+z\leq e^{z}$ for $z\in \mathbb{R}$. To see $(c)$, observe that for $p\leq 0.38$, we have, $\frac{\sqrt{p_i}}{1-p_i}\leq 1$. Then, using the inequality $e^{z}\leq 1+2z$ for $z\in[0,1]$, gives us the final inequality.
		\smallskip
		
		\noindent \emph{Proof of $(ii)$:} 
		From part $(i)$ we have, 
		\[F^{\umi}_{\mcala}(x) = \sum_{j=1}^k \left[p_{it_j} \prod_{\ell=1}^j (1-p_{it_\ell})\right] = \left(1\pm O\left(\sqrt{p_i}\right)\right) \sum_{j=1}^k p_{it_j} e^{-p(j)}
		\quad \forall x\in\left[0,\frac{1}{\sqrt{p_i}}\right].\]
		Notice that $ \sum_{j=1}^k p_{it_j} e^{-p(j)}=(1\pm O(p_i)) \int_0^x e^{-z} dz$, which establishes part $(ii)$ for $x\leq \frac{1}{\sqrt{p_i}}$. Consequently, we have, $F^{\umi}_{\mcala}(\frac{1}{\sqrt{p_i}})\geq (1-O(\sqrt{p_i}))$.
		
		For $x>\frac{1}{\sqrt{p_i}}$, we have 
		$F^{\umi}_{\mcala}(\frac{1}{\sqrt{p_i}}) \leq F^{\umi}_{\mcala}(x)\leq 1$. Then, using $F^{\umi}_{\mcala}(\frac{1}{\sqrt{p_i}})\geq (1-O(\sqrt{p_i}))$, we get, $ F^{\umi}_{\mcala}(x) = (1\pm 
		O(\sqrt{p_i}))\,\, (1-e^{-x})\quad \forall x\geq 0.$
		
		\hfill	\Halmos\end{proof}
	
	\begin{repeatlemma}[Lemma \ref{funcapx}.]
		Consider a deterministic non-anticipative algorithm $\mcala$, resouce $i\in I$, threshold vector $\umi\in [0,1]^{n-1}$, a non-decreasing {\color{black} integrable function $h: \mathbb{R}^+ \to \mathbb{R}^+$ such that,
		\begin{enumerate}[(i)]
			\item $h(0)=0$,
			\item $h(x)=h(\frac{1}{\sqrt{p_i}})\,\, \forall x\geq \frac{1}{\sqrt{p_i}}$.
		\end{enumerate}} 
	For $p\leq 0.38$, we have,
		\[ E_{u_i} [h\left(G^{\umi}_{\mcala} (u_i)\right)] = (1\pm O(\sqrt{p_i}))\, E_{b \sim Exp(1)} [h(b)]. \]
	\end{repeatlemma}
	\begin{proof}{Proof.} Let $p(0)=0$ and recall that $p(\ell)=p(m^{\umi}_{\mcala})+(\ell-m^{\umi}_{\mcala}) p_i\quad \forall \ell \geq m^{\umi}_{\mcala}$. For simplicity, let $\Delta h(k)=h(p(k))-h(p(k-1))\quad \forall k\geq 1$. 
	{\color{black}	Using the fact that $h(x)=h(\frac{1}{\sqrt{p_i}})$ for all $x\geq 1/\sqrt{p_i}$, we have, $\Delta h(k)=0$ for every $k$ such that $p(k-1)\geq \frac{1}{\sqrt{p_i}}$.
		\begin{eqnarray*}
			E_{u_i} [h\left(G^{\umi}_{\mcala} (u_i)\right)] &=  &E_{b\sim F^{\umi}_{\mcala}} [h(b)], \\
			&= &\sum_{k\geq 1} (1-F^{\umi}_{\mcala}(p(k-1))) \Delta h(k), \\
			&= & (1\pm O(\sqrt{p_i})) \sum_{k\geq 1} e^{-p(k-1)} \Delta h(k),\\
				&= & (1\pm O(\sqrt{p_i})) \sum_{k\geq 1} \left(e^{-p(k-1)}-e^{-p(k)}\right)\, h(p(k)) ,\\
					&= & (1\pm O(\sqrt{p_i})) \sum_{k\geq 1} \left(h(p(k))\, \int_{p(k-1)}^{p(k)}e^{-x}dx \right),
		\end{eqnarray*}
		here the first and second equation follow from the definition of $F^{\umi}_{\mcala}$. The third equation follows from Lemma \ref{expapprox} $(i)$ and the fact that $\Delta h(k)=0$ for every $k$ such that $p(k-1)\geq \frac{1}{\sqrt{p_i}}$. 
		
		Now, to prove the main claim it suffices to show that,
		\begin{equation}	 \label{step}
			 \sum_{k\geq 1} \left(h(p(k))\, \int_{p(k-1)}^{p(k)}e^{-x}dx \right)
		=  (1\pm O(p_i)) \int_0^{+\infty} e^{-x}\, h(x)\, dx.
	\end{equation}
	Observe that for every non-decreasing and integrable function $h(\cdot)$,
		\[  \int_{p(k-1)}^{p(k)}e^{-x}\,h(x)\, dx\,\leq\, h(p(k))\, \int_{p(k-1)}^{p(k)}e^{-x}dx\,\leq \, e^{p_i} \int_{p(k-1)+p_i}^{p(k)+p_i}e^{-x}\,h(x)\, dx\, \quad \forall k\geq 1.\] 
		Summing up the inequalities above for $k\geq 1$ and using the fact that $e^{p_i}\leq 1+2p_i$ for $p_i\in[0,1]$, completes the proof.}
%
		\hfill\Halmos	
	\end{proof}

\begin{repeatlemma}[Lemma \ref{determ}.]
	The optimal fully offline and clairvoyant algorithms are deterministic and given by a (possibly exponential time) dynamic program.
\end{repeatlemma} 
\begin{proof}{Proof.}
	Let us start with clairvoyant benchmark. The state space is given by the set of arrivals matched and the success/failure of matches so far. The decision space of clairvoyant is simply the matching decision for the next arrival in the sequence. Let $V(t,S_t)$ denote the optimal value-to-go given that at arrival $t$, the set of remaining resources is $S_t$. Clearly,
	\[V(t,S_t)=\max_{i\mid i\in S_t, (i,t)\in E} p_{it}\big(r_i+ V(t+1,S_t\backslash\{i\})\big) + (1-p_{it}) V(t+1,S_t). \]
Therefore,	regardless of the values of $V(t+1,S_t)$ and $V(t+1,S_t\backslash\{i\})$, the optimal decision at arrival $t$ is deterministic. In case of fully offline, the action space also includes the order in which arrivals are visited. It is easy to see that a similar argument applies to fully offline.
	\hfill\Halmos 

\end{proof}
		\subsection{Missing Proofs from Section \ref{sec:fsb}} \label{appx:fsb}
		\begin{repeatlemma}[Lemma \ref{trunc}.]
			For every $i\in I$,	let $\opt_i$ and $\topt_i$ denote the expected total revenue from matching resource $i$ in \opt\ and in $\topt$ respectively. Then, we have, $\topt_i \geq (1- O(\sqrt{p})) \opt_i\,\,\, \forall i\in I$. 
		\end{repeatlemma}
		\begin{proof}{Proof.}
			Consider an arbitrary resource $i\in I$ and fix sample path $\umi\in[0,1]^{n-1}$. Let $p(m^{\umi}_\opt)$ and $p(m^{\umi}_{\topt})$ denote the effort threshold of $i$ in \opt\ and $\topt$ respectively. By definition of $\topt$, we have, $p(m^{\umi}_{\topt})=\min\left\{p(m^{\umi}_\opt),\frac{1}{\sqrt{p}}\right\}$. Let $p_1$ denote the probability that $i$ is successfully matched in \opt\ (conditioned on $\umi$). Similarly, let $p_2$ denote the probability that $i$ is successfully matched in $\topt$. It suffices to show that $p_2\geq (1-O(\sqrt{p})) p_1$. 
			
			When $p(m^{\umi}_{\topt})=p(m^{\umi}_\opt)$, we have, $p_2=p_1$, since \opt\ and $\topt$ match $i$ to the same set of arrivals. If $p(m^{\umi}_{\topt})<p(m^{\umi}_\opt)$, then we have, $p(m^{\umi}_{\topt})=\frac{1}{\sqrt{p}}$. Using Corollary \ref{basiceq}$(ii)$ and Lemma \ref{expapprox}$(ii)$, we have that, $p_2\geq (1-O(\sqrt{p_i}))\, e^{-\frac{1}{\sqrt{p}}}\geq 1-O(\sqrt{p})$. Since $p_1\leq 1$, we have the desired. 
			
			\hfill\Halmos
			
		\end{proof}
\section{Comparison with Configuration LP in \cite{huang}}\label{appx:huang}

Recall that \cite{huang} show a 0.576 competitive ratio guarantee against the expectation LP for vanishingly small probabilities. The use a primal-dual approach with a new LP that approximates the expectation LP for small probabilities. The dual of the LP in Section 2.2 of \cite{huang}, has the following constraints (paraphrased in our notation),
	\[\theta^{}_i + \sum_{t\in N_i} \lambda^{}_t \geq \alpha r_i \min\Big\{1,\sum_{t\in N_i}p_{it}\Big\}\quad \forall i,N_i\subseteq \{t\mid (i,t)\in E \}. \]
These constraints differ from inequalities \eqref{dual1} in our LP free framework in two important ways: (i) For any given $i$, the above constraints are imposed for all subsets of edges incident on $i$. In contrast, inequalities \eqref{dual1} are imposed over all possible \emph{sample paths} $\umi$. The dependence on sample paths allows us 
to define our candidate solution as a function of coupled randomness in \alg\ and \opt. 
(ii) The quantity $\min\Big\{1,\sum_{t\in N_i}p_{it}\Big\}$ is, in general, a loose upper bound on $\opt_i$, even for small edge probabilities. Recall that in \eqref{dual1}, we approximate $\opt_i$ with $r_i\, E_{\umi}[(1-e^{p(m^{\umi}_{\opt})})]$.

 \end{APPENDICES}


\ACKNOWLEDGMENT{We would like to thank anonymous referees for their valuable feedback. We also thank Will Ma for a helpful discussion on the results in \cite{deb}.  
	This research was partially supported by NSF Grant CMMI 1636046 and NSF CAREER Award CMMI 1351838.}


\bibliographystyle{informs2014} 
\bibliography{bib} 



\end{document}